\documentclass[11pt]{article}
\usepackage[utf8]{inputenc}
\usepackage{bbm}
\usepackage{amsmath}
\usepackage{csquotes}
\usepackage{mathtools}
\usepackage{amsthm}
\usepackage{amssymb}
\usepackage{pdflscape}
\usepackage{units}
\usepackage{amsfonts}
\usepackage[english]{babel}
\usepackage{xfrac}
\usepackage{algorithm}
\usepackage{algpseudocode}
\usepackage[hang]{footmisc}
\usepackage{lipsum}
\usepackage{enumitem}
\usepackage{hyperref}
\hypersetup{colorlinks=true,allcolors=blue}
\usepackage{multirow}
\usepackage{graphicx,xcolor,textpos}
\usepackage{natbib}
\usepackage{csquotes}
\usepackage{emptypage}
\usepackage[toc,page]{appendix}
\usepackage[toc,acronym,nonumberlist]{glossaries}
\usepackage{diagbox}
\usepackage{caption}
\usepackage{subcaption}
\usepackage{float}
\usepackage{listings}
\usepackage{breqn}
\usepackage{accents}

\DeclareMathOperator{\EX}{\mathbb{E}}

\newcommand{\diff}{\mathop{}\!d}

\newcommand{\indep}{\perp \!\!\! \perp}

\newcommand*{\QEDB}{\null\nobreak\hfill\ensuremath{\square}}%

\newtheorem{theorem}{Theorem}
\newtheorem{lemma}{Lemma}
\newtheorem{proposition}{Proposition}

\begin{document}

\title{Tests of exogeneity in duration models with censored data}

\author{$\text{Gilles Crommen}^1$ \\ \url{gilles.crommen@kuleuven.be} \and $\text{Jean-Pierre Florens}^{2}$ \\ \url{jean-pierre.florens@tse-fr.eu} \and $\text{Ingrid Van Keilegom}^1$ \\ \url{ingrid.vankeilegom@kuleuven.be}}

\date{$^1$ORSTAT, KU Leuven, Naamsestraat 69, 3000 Leuven, Belgium \\ $^2$Toulouse School of Economics, Esplanade de l’Université 1, 31080 Toulouse, France}

\maketitle

\begin{abstract}
\noindent Consider the setting in which a researcher is interested in the causal effect of a treatment $Z$ on a duration time $T$, which is subject to right censoring.
We assume that $T=\varphi(X,Z,U)$, where $X$ is a vector of baseline covariates, $\varphi(X,Z,U)$ is strictly increasing in the error term $U$ for each $(X,Z)$ and $U\sim \mathcal{U}[0,1]$. Therefore, the model is nonparametric and nonseparable. We propose nonparametric tests for the hypothesis that $Z$ is exogenous, meaning that $Z$ is independent of $U$ given $X$. The test statistics rely on an instrumental variable $W$ that is independent of $U$ given $X$. We assume that $X,W$ and $Z$ are all categorical. Test statistics are constructed for the hypothesis that the conditional rank $V_T= F_{T \mid X,Z}(T \mid X,Z)$ is independent of $(X,W)$ jointly. Under an identifiability condition on $\varphi$, this hypothesis is equivalent to $Z$ being exogenous. However, note that $V_T$ is censored by $V_C =F_{T \mid X,Z}(C \mid X,Z)$, which complicates the construction of the test statistics significantly. We derive the limiting distributions of the proposed tests and prove that our estimator of the distribution of $V_T$ converges to the uniform distribution at a rate faster than the usual parametric $n^{-1/2}$-rate. We demonstrate that the test statistics and bootstrap approximations for the critical values have a good finite sample performance in various Monte Carlo settings. Finally, we illustrate the tests with an empirical application to the National Job Training Partnership Act (JTPA) Study.
\end{abstract}

\maketitle

\section{Introduction}

We consider the setting in which a researcher is interested in identifying the causal effect of a treatment variable $Z$ on a duration of interest $T$. Note that $Z$ could be a policy intervention, a new medical treatment or other forms of exposure. A central challenge in this context arises from the problem of endogeneity, which refers to $Z$ not being independent of the error term in the structural model for $T$. Endogeneity can result from a variety of sources such as unobserved confounders, sample selection, measurement error and noncompliance \citep{rubin1974estimating,HeckmanJamesJ.1979SSBa,angrist1996identification}. When $Z$ is endogenous, standard estimation methods can be severely biased. Because much applied work relies on observational data, addressing endogeneity should be a necessary step in any empirical analysis.

When the treatment $Z$ is exogenous, meaning independent of the error term in the structural model for the outcome $T$, estimating its causal effect is a well-posed problem for which standard estimation approaches are well-suited. By contrast, when $Z$ is endogenous, estimation of the treatment effect becomes more challenging. A common approach to deal with endogeneity is the use of instrumental variable (IV) methods. This approach relies on external sources of variation (instruments) that influence $Z$, but are independent of the unobserved heterogeneity affecting both $Z$ and $T$. However, with nonparametric IV methods, the estimation problem often takes the form of an ill-posed inverse problem \citep{newey2003instrumental,chernozhukov2005iv,darolles2011nonparametric,blundell2013nonparametric}. This means that nonparametric estimation with instruments is more convoluted and unstable compared to standard estimation methods such as ordinary least squares or a general quantile estimator. As a result, researchers often focus on simpler estimands, such as averages, rather than attempting to recover full functional relationships. These issues highlight the importance of having a statistical test for the hypothesis that $Z$ is exogenous.

We consider the following nonparametric nonseparable model
\begin{equation}\label{nonseperablemodel}
    T=\varphi(X,Z,U),
\end{equation}
where the map $u \mapsto \varphi(X,Z,u)$ is strictly increasing for each $(X,Z)$ and $U \sim \mathcal{U}[0,1]$. The variable $T$ is a duration of interest, $Z$ is a possibly endogenous treatment variable and $X$ represents a vector of baseline covariates. Note that if $\varphi(X,Z,\cdot)$ is strictly increasing for each $(X,Z)$, the rank invariance assumption as described by \cite{chernozhukov2005iv} is implied \citep{dong2018testing}. The rank invariance assumption is a condition on the potential outcomes and implies that individuals have the same unobserved rank when their treatment changes. Further, we will suppose that there exists an instrument $W$ that is independent of $U$ given $X$. We also allow for the duration time $T$ to be censored by a right censoring time $C$. Therefore, we only observe the follow-up time $Y=\min\{T,C\}$ and the censoring indicator $\Delta=\mathbbm{1}(Y=T)$ alongside $X,W$ and $Z$. We will assume that $X,W$ and $Z$ are categorical variables with finite support to facilitate the construction of the test statistics. Moreover, it is assumed that $C$ is independent of $T$ and $W$ jointly given $X$ and $Z$. The goal of this paper is to develop nonparametric tests for the hypothesis that $Z$ is exogenous, meaning that $Z$ is independent of $U$ conditional on $X$, without having to estimate the function $\varphi$.

There are two main approaches to test for exogeneity in a nonparametric nonseparable model. The first is based on estimating the function $\varphi$ under the identification assumptions of the nonparametric IV model, and comparing the estimate to general regression or quantile estimates. The second approach is to estimate the residuals under the exogeneity assumption and verify if the IV condition is satisfied. We will follow the second approach to avoid the difficulties associated with nonparametric IV estimation. This puts us in a similar framework to \cite{feve2018estimation}, who developed a nonparametric test for exogeneity based on an estimator of the distribution of the conditional rank $F_{T\mid Z}(T\mid Z)$. Contrary to \cite{feve2018estimation}, we allow for $T$ to be subject to right censoring. Moreover, we also allow for the inclusion of baseline covariates in model \eqref{nonseperablemodel}. These generalizations introduce significant challenges during the construction of the test statistics. Note that if $T$ is right censored by $C$, this implies that $V_T=F_{T\mid X,Z}(T\mid X,Z)$ is right censored by $V_C=F_{T\mid X,Z}(C\mid X,Z)$. Therefore, conditional on $F_{T\mid X,Z}$, we only observe $V = \min\{V_T,V_C\}$ and $\Delta=\mathbbm{1}(Y=T)$. Because of this, also the estimated  $\widehat{V}_T=\widehat{F}_{T\mid X,Z}(T\mid X,Z)$ is subject to right censoring, which complicates the estimation of the distribution of $V_T$ significantly. Whereas \cite{feve2018estimation} permits $W$ and $Z$ to be continuous, we restrict $X,W$ and $Z$ to be categorical. This restriction has the advantage that we do not need to smooth over the covariates, treatment and duration time when estimating $F_{T\mid X,Z}$. Under suitable identification conditions on $\varphi$, we show that the independence between $V_T$ and $(X,W)$ is equivalent to $Z$ being independent of $U$ given $X$. Other tests of exogeneity for censored duration outcomes have been proposed in the literature by \cite{smith1986exogeneity} and \cite{rivers1988limited} among others. However, these tests assume fixed and fully observed censoring times in Tobit models. To the best of our knowledge, no other test for exogeneity in a nonparametric nonseparable model that allows for random right censoring is available in the literature.

The nonparametric nonseparable model is a common framework for analyzing the effects of endogenous treatments on censored duration outcomes. Several nonparametric contributions that focus on estimating local average treatment effects under a monotonicity assumption include \cite{frandsen2015treatment}, \cite{sant2016program} and \cite{blanco2020bounds}, while \cite{chernozhukov2015quantile} and \cite{beyhum2022nonparametric} address identification and estimation of population-level treatment effects under a type of rank invariance assumption. See \cite{wuthrich2020comparison} for a comparison of rank invariance to monotonicity for the identification of quantile treatment effects. Another strand of the literature achieves identification under separability assumptions \citep{abbring2003nonparametric,abbring2005social,bijwaard2005correcting}. Some semiparametric approaches include \cite{tchetgen2015instrumental}, \cite{li2015instrumental}, \cite{chan2016reader} \cite{kianian2021causal}, \cite{beyhum2024instrumental} and \cite{tedesco2025instrumental} among others. 

The remainder of the paper is organized as follows. Section \ref{sectests} discusses under which conditions exogeneity of $Z$ is equivalent to the independence between $V_T$ and $(X,W)$, how the test statistics are constructed and explains the estimation procedure. The test statistics' limiting distributions and two bootstrap approaches to approximate the critical values are given in Section \ref{secdisttets}. The finite-sample performance is investigated in Section \ref{secfinitesample} through Monte Carlo simulations and Section \ref{secempapp} provides an empirical application to the National Job Training Partnership Act (JTPA) Study. Section \ref{secconclusion} discusses possible extensions and practical considerations. The technical proofs can be found in the Appendix.

\section{The test statistics}\label{sectests}
In this paper, we develop nonparametric tests for the null hypothesis that $$H_0 : V_T \indep (X,W),$$ where $
V_T = F_{T \mid X,Z}(T \mid X,Z)$ with $F_{T \mid X,Z}(t \mid x,z) = \mathbbm{P}(T \leq t\mid X=x,Z=z)$, $X$ is a vector of baseline covariates and $W$ is an instrumental variable such that $W \indep U\mid X$ with $\indep$ denoting statistical independence. Because $T$ is right censored by $C$, $V_T$ is right censored by $V_C = F_{T \mid X,Z}(C \mid X,Z)$. Therefore, conditionally on $F_{T \mid X,Z}$, we only observe $V=F_{T \mid X,Z}(Y \mid X,Z)=\min\{V_T,V_C\}$ and $\Delta = \mathbbm{1}(Y=T)$. Note that $T$ and $V_T$ both being subject to right censoring introduces significant challenges in constructing the test statistics. We now introduce the following model condition:
\begin{enumerate}[label=(A0),resume,left=0.25\leftmargin]
\item Model \eqref{nonseperablemodel} is identified. Moreover, $\varphi(X,Z,U)$ is strictly increasing in $U$ for each $(X,Z)$, $U \sim \mathcal{U}[0,1]$ and $W$ is independent of $U$ conditionally on $X$.\label{globidentass}
\end{enumerate}
The model being identified means that if
$$ T=\varphi_1(X,Z,U_1) = \varphi_2(X,Z,U_2), $$ 
and $U_1$ and $U_2$ are both independent of $W$ conditionally on $X$ and uniform on $[0,1]$, then
$$ U_1=U_2 \text{ a.s.\ and } \varphi_1=\varphi_2. $$ 
Note that $U\sim \mathcal{U}[0,1]$ can be assumed without loss of generality as long as $U$ is continuously distributed and has positive density on its support. The equivalence between $V_T \indep (X,W)$ and $Z$ being exogenous, meaning that $Z \indep U \mid X$, is established by the following proposition.
\begin{proposition}\label{proptest}
    Assume \ref{globidentass}, then  $Z \indep U \mid X$ { $\Longleftrightarrow$} {$W \indep V_T \mid X$} { $\Longleftrightarrow$} $V_T \indep (X,W)$.
\end{proposition}
\begin{proof}
    If $Z \indep U \mid X$, clearly $U = V_T$ a.s. and hence $W \indep V_T \mid X$. Since by construction $V_T \indep X$, we have that $V_T \indep (X,W).$ On the other hand, if $V_T \indep (X,W)$, it is clear that $W \indep V_T \mid X$ and by noticing that $T = F^{-1}_{T \mid X,Z}(V_T \mid X,Z)$ we can conclude that $U = V_T$ a.s. by the identification of model \eqref{nonseperablemodel}. Since $V_T$ is independent of $Z$ by construction, clearly $Z \indep U \mid X$.
\end{proof}
Proposition \ref{proptest} shows that while our independence test between $V_T$ and $(X,W)$ does not rely on Assumption \ref{globidentass}, the independence property can be interpreted as an exogeneity property under this assumption. 
The identification of model \eqref{nonseperablemodel}, which is nonparametric and nonseparable, using instrumental variables is a complex issue that has been studied by \cite{chesher2003identification}, \cite{chernozhukov2005iv}, \cite{chernozhukov2007instrumental} and \cite{chen2014local} among others.
A sufficient condition for the identification of model \eqref{nonseperablemodel} can be found in Appendix A of \cite{feve2018estimation}, which is a completeness type condition that has been proven to be nontestable \citep{canay2013testability}.

Under $H_0$, it is clear that \begin{equation}
    \mathbbm{P}(V_T \leq v, X = x, W =w) - \mathbbm{P}(V_T \leq v)\mathbbm{P}( X = x, W = w),\label{jointind}\end{equation} is equal to zero for all $ (v,x,w)$. Equivalently, we have that \begin{equation}\mathbbm{P}(V_T \leq v\mid X = x,W  = w)-\mathbbm{P}(V_T \leq v),\label{condind}\end{equation} is equal to zero for all $v$ and all $(x,w)$ for which $\mathbbm{P}(X=x,W=w) > 0$.
    While it is possible to construct our test statistics based on an estimator of \eqref{jointind}, $V_T$ being subject to right censoring by $V_C$ substantially complicates estimation of $\mathbbm{P}(V_T \leq v, X = x, W =w)$. Nonetheless, possible estimators for this joint distribution have been proposed by \cite{stute1993consistent} and \cite{akritas1994nearest}. The estimator proposed by \cite{stute1993consistent} is simple to implement, but it would require the additional assumption that $\mathbbm{P}(V_T \leq V_C \mid X,W,V_T)=\mathbbm{P}(V_T \leq V_C \mid V_T)$. On the other hand, the approach proposed by \cite{akritas1994nearest} requires the estimation of $\mathbbm{P}(V_T \leq v\mid X = x,W  = w)$ as a first step to estimate $\mathbbm{P}(V_T \leq v, X = x, W =w)$. Therefore, we will construct our test statistics based on a nonparametric estimator of $\eqref{condind}$. To facilitate the construction of the test statistics, we will make the following key assumptions:
\begin{enumerate}[label=(A\arabic*),left=0.25\leftmargin]
\item $X$, $W$ and $Z$ take values in the finite sets $\mathcal{R}_X = \{x_1,\dots,x_{d_x}\}, \mathcal{R}_W=\{w_1,\dots,w_{d_w}\}$ and $\mathcal{R}_Z=\{z_1,\dots,z_{d_z}\}$ respectively.\label{Asupport}
\item $C \indep (T,W) \mid X,Z$.\label{AindepTC}
\end{enumerate}
Even though assumption \ref{Asupport} might be restrictive in some settings, it will allow us to develop a tractable, fully nonparametric approach to test $H_0$. Moreover, it has the advantage that we do not need to smooth over the covariates, treatment and duration time when estimating $F_{T\mid X,Z}$. Note that under Assumption \ref{Asupport}, a necessary condition for Assumption \ref{globidentass} to hold is that $d_w \geq d_z$. Possible extensions to allow for continuous covariates, instruments and treatments are discussed in Section \ref{secconclusion}. Further, it follows that Assumption \ref{AindepTC} is equivalent to assuming that $(i)$ $C \indep W \mid X,Z$ and $(ii)$ $T \indep C \mid X,W,Z$. Clearly, $(i)$ implies that after conditioning on $(X,Z)$, the instrument should not affect the censoring time $C$. In particular, if $W$ indicates being randomized to a treatment or control group, this means that, conditional on the baseline covariates $X$ and treatment participation $Z$, being assigned to the treatment or control group should not influence the censoring time. Practitioners can check for violations of $(i)$ by comparing the censoring rates for different levels of $W$ conditional on $(X,Z)$. On the other hand, $(ii)$ implies that the duration and censoring times are independent given the baseline covariates, treatment assignment and treatment participation. Even though this assumption is widely adopted in the literature on censored duration outcomes, it is, in general, nontestable \citep{TsiatisA.1975ANAo} and its violation can lead to biased estimates \citep{moeschberger_consequences_1984}. Note that when $C$ is observed in addition to $Y$ and $\Delta$, \cite{frandsen2019testing} proposed a test for censoring point independence. In Section \ref{secconclusion}, we outline a possible extension to our approach that allows for some forms of dependence between $T$ and $C$ after conditioning on $(X,W,Z)$.

\subsection{Estimation}\label{secesttest}

For all $i=1,\dots,n$, define
$$
\widehat{V}_i = \widehat{F}_{T \mid X,Z}(Y_i \mid X_i,Z_i),
$$
where $\widehat{F}_{T \mid X,Z}$ is a conditional Kaplan-Meier estimator, that is,
$$
\widehat{F}_{T \mid X,Z}(t \mid x,z) =1-\prod_{i : Y_{(i)} \leq t} \left(1-\frac{d_{(i)}(x,z)}{r_{(i)}(x,z)}\right),
$$
with $Y_{(1)},\dots,Y_{(m)}$ the $m$ ordered distinct follow-up times, the number of events $d_{(i)}(x,z) = \sum_{j=1}^n\mathbbm{1}(Y_j = Y_{(i)},\Delta_j = 1,X_j=x, Z_j=z)$ and the risk set $r_{(i)}(x,z)= \sum_{j=1}^n\mathbbm{1}(Y_j \geq Y_{(i)},X_j=x, Z_j=z).$ Because $\widehat{F}_{T \mid X,Z}( \cdot \mid x,z)$ is a monotonically increasing function, we have that $$\widehat{V}_i = \widehat{F}_{T \mid X,Z}(\min\left\{T_i,C_i\right\} \mid X_i,Z_i) = \min\left\{\widehat{F}_{T \mid X,Z}(T_i \mid X_i,Z_i),\widehat{F}_{T \mid X,Z}(C_i \mid X_i,Z_i)\right\},$$
such that $\widehat{V}_{T_i} = \widehat{F}_{T \mid X,Z}(T_i \mid X_i,Z_i)$ is censored by $\widehat{V}_{C_i} = \widehat{F}_{T \mid X,Z}(C_i \mid X_i,Z_i)$ with the same censoring indicator $\Delta_i=\mathbbm{1}(Y_i=T_i)$. It is important to note that, since $V_T \indep X,Z$ by construction, Assumption \ref{AindepTC} implies that $V_T \indep V_C$.
Therefore, we can estimate the distribution function of $V_T=F_{T \mid X,Z}(T \mid X,Z)$ by a Kaplan-Meier estimator, plugging-in $\{\widehat{V}_i\}_{i=1,\dots,n}$ as the observed follow-up times. Specifically, let
$$
\widehat{F}_{\widehat{V}_T}(v)= 1-\prod_{i:\widehat{V}_{(i)} \leq v} \left(1-\frac{\widehat{d}_{(i)}}{\widehat{r}_{(i)}}\right),
$$
with $\widehat{V}_{(1)},\dots,\widehat{V}_{(k)}$ the $k$ ordered distinct estimated conditional ranks, $\widehat{d}_{(i)}=\sum_{j=1}^n\mathbbm{1}(\widehat{V}_j = \widehat{V}_{(i)},\Delta_j = 1)$ and $\widehat{r}_{(i)}=\sum_{j=1}^n\mathbbm{1}(\widehat{V}_j \geq \widehat{V}_{(i)})$. Lastly, it is important to note that we cannot estimate $F_{V_T\mid X,W}(v\mid x, w)=\mathbbm{P}(V_T \leq v\mid  X = x,W = w)$ by a conditional Kaplan-Meier estimator, since our assumptions do not imply $V_T \indep V_C \mid X,W$. However, Assumption \ref{AindepTC} does imply that $V_T \indep V_C \mid X,Z,W$. Therefore, we can estimate $F_{V_T\mid X,W,Z}(v\mid x, w,z)=\mathbbm{P}(V_T \leq v\mid X=x,W=w,Z=z)$ by a conditional Kaplan-Meier estimator and marginalize out $Z$, that is,
$$
\widehat{F}_{\widehat{V}_T\mid X,W}(v\mid x, w) =\sum_{z \in\mathcal{R}_Z}\widehat{F}_{\widehat{V}_T\mid X,W,Z}(v\mid x,w,z)\widehat{p}_{Z\mid X,W}(z\mid x,w),$$ 
where
$$
\widehat{p}_{Z\mid X,W}(z\mid x,w) = \frac{ \sum_{i=1}^n\mathbbm{1}(X_i=x,W_i=w,Z_i=z)}{\sum_{i=1}^n\mathbbm{1}(X_i=x,W_i=w)},
$$ is an estimator of $p_{Z\mid X,W}(z\mid x,w)=\mathbbm{P}(Z=z\mid X=x,W=w)$ and $$
\widehat{F}_{\widehat{V}_T\mid X,W,Z}(v\mid x, w,z) =1-\prod_{i:\widehat{V}_{(i)} \leq v} \left(1-\frac{\widehat{d}_{(i)}(x,w,z)}{\widehat{r}_{(i)}(x,w,z)}\right),
$$ with $\widehat{d}_{(i)}(x,w,z)= \sum_{j=1}^n\mathbbm{1}(\widehat{V}_j = \widehat{V}_{(i)},\Delta_j = 1,X_j = x, W_j = w, Z_j = z)$ and $\widehat{r}_{(i)}(x,w,z) = \sum_{j=1}^n\mathbbm{1}(\widehat{V}_j \geq \widehat{V}_{(i)}, X_j = x, W_j = w, Z_j = z).$
Finally, let $$\widehat{D}(v,x,w)=\widehat{F}_{\widehat{V}_T\mid X,W}(v\mid x,w) - \widehat{F}_{\widehat{V}_T}(v),$$
be an estimator for \eqref{condind}. To check the null hypothesis, we can compute the following Kolmogorov-Smirnov statistic:
\begin{equation}\label{eqteststatisticKS}
T^{KS}_{n}= \sqrt{n}\sup_{v \in \mathcal{I},(x,w) \in \mathcal{R}_{X,W}}\lvert \widehat{D}(v,x,w) \rvert,    
\end{equation}
or the following weighted Cramér-von Mises statistic:
\begin{equation}\label{eqteststatisticCM}
T^{CM}_{n}= n \sum_{(x,w) \in \mathcal{R}_{X,W}} \widehat{\pi}(x,w) \int_{\mathcal{I} } \widehat{D}^2(v,x,w) \diff\widehat{F}_{\widehat{V}_T\mid X,W}(v\mid x,w),    
\end{equation}
where $\widehat{\pi}(x,w)$ is a specified weight function, $\mathcal{R}_{X,W} = \{(x,w) \in \mathcal{R}_X \times \mathcal{R}_W : \mathbbm{P}(X=x,W=w) > 0 \}$ and $\mathcal{I} \subseteq [0,1-\gamma]$ for some small $\gamma > 0$ that will be defined by Assumption \ref{supportass}. For example, one may take the weight function $\widehat{\pi}(x,w)$ of the Cramér-von Mises statistic to be constant or to equal the empirical proportion $n^{-1}\sum_{i=1}^n\mathbbm{1}(X_i=x,W_i=w)$. Note that $\gamma$ being strictly positive is only necessary for the asymptotic theory. In Section \ref{secfinitesample}, we will set $\gamma = 0$ and show that the test still performs well using Monte Carlo simulations.

\section{Distribution of test statistics}\label{secdisttets}
Before stating the main theorems regarding the test statistics' limiting distributions, we need to introduce some more notation and assumptions. Firstly, let 
\begin{align*}
 & p_{X,W}(x,w)=\mathbbm{P}(X=x,W=w), \quad p_{W\mid V_T,X,Z}(w\mid v,x,z) = \mathbbm{P}(W=w\mid V_T=v,X=x,Z=z), \\ &S_{V_T}(v) = 1 - F_{V_T}(v),  \quad S_{V_T\mid X,W,Z}(v\mid x,w,z) = 1- F_{V_T\mid X,W,Z}(v\mid x,w,z), \\ & S_{V,1\mid X,W,Z}(v\mid x,w,z) = \mathbbm{P}(F_{T \mid X,Z}(Y \mid X,Z) > v,\Delta =1 \mid X=x,W=w,Z=z), \\ &    S_{V\mid X,W,Z}(v\mid x,w,z) = \mathbbm{P}(F_{T \mid X,Z}(Y \mid X,Z) > v \mid X=x,W=w,Z=z),    
\end{align*}
and similarly for $S_{V,1\mid X,Z}$ and $S_{V\mid X,Z}$. Moreover, let  $$\zeta_i(v,x,w,z) = \int^v_0\frac{N_{i,x,w,z}(u)dS_{V,1\mid X,W,Z}(u\mid x,w,z)}{S_{V\mid X,W,Z}(u\mid x,w,z)^2}-\int^v_0\frac{dN^1_{i,x,w,z}(u)}{S_{V\mid X,W,Z}(u\mid x,w,z)},$$ and $$\xi_i(u,x,z) = \int^u_0\frac{N_{i,x,z}(s)dS_{V,1\mid X,Z}(s\mid x,z)}{S_{V\mid X,Z}(s\mid x,z)^2}-\int^u_0\frac{dN^1_{i,x,z}(s)}{S_{V\mid X,Z}(s\mid x,z)},$$
with $N_{i,x,w,z}(v)= \mathbbm{1}(V_i\geq v,X_i=x,W_i = w,Z_i = z)$, $N^1_{i,x,w,z}(v)=\mathbbm{1}(V_i\geq v, \Delta_i = 1,X_i=x,W_i=w,Z_i = z)$ and similarly for $N_{i,x,z}(v)$ and $N^1_{i,x,z}(v)$. Lastly, we will need the following regularity assumptions:
\begin{enumerate}[label=(A\arabic*),resume,left=0.25\leftmargin]
\item For all $(x,w,z) \in \mathcal{R}_X \times \mathcal{R}_W \times \mathcal{R}_Z $, the conditional distribution function $F_{Y \mid X,W,Z}(y\mid x,w,z)=\mathbbm{P}(Y\leq y \mid X=x,W=w,Z=z)$, with $f_{Y\mid X,W,Z}(y\mid x,w,z)$ the respective conditional density function, is twice continuously differentiable with respect to $y$, such that
 $$ \sup_{y}\left\lvert \frac{\partial^k}{\partial y^k}F_{Y\mid X,W,Z}(y\mid x,w,z)\right\rvert < \infty,
$$
for $k=1,2$. Moreover, for all $(x,w,z) \in \mathcal{R}_X \times \mathcal{R}_W \times \mathcal{R}_Z $, we have that
$$
\sup_{0 \leq F_{T\mid X,Z}(y\mid x,z)\leq1-\gamma}\frac{f_{Y\mid X,W,Z}(y\mid x,w,z)}{f_{T\mid X,Z}(y\mid x,z)} < \infty,
$$where $\gamma >0$ is such that $F_{Y \mid X,W,Z}(F^{-1}_{T\mid X,Z}(1-\gamma\mid x,z)\mid x,w,z)<1$ for all $(x,w,z) \in \mathcal{R}_X \times \mathcal{R}_W \times \mathcal{R}_Z $ and $f_{T\mid X,Z}(y\mid x,z)$ is the conditional density function of $T$ given $X=x$ and $Z=z$. \label{supportass}
\item We have that
$$
\sup_{(x,w) \in \mathcal{R}_{X,W}}\left\lvert \widehat{\pi}(x,w) -\pi(x,w) \right\rvert \xrightarrow{a.s.}  0, 
$$ for some function $\pi(x,w)$.\label{ascoonvpi}
\end{enumerate}
We now have the following result.
\subsection{Limiting distributions}\label{asymptoticsection}
\begin{theorem}\label{TH1}
Under Assumptions \ref{Asupport} - \ref{supportass}, we have that
     \begin{align*}
     \widehat{D}(v,x,w) = n^{-1}\sum_{i=1}^n\frac{\omega_i(v,x,w)}{p_{X,W}(x,w)} +F_{V_T\mid X,W}(v\mid x,w) -v +o_p(n^{-1/2}),  
    \end{align*}
    uniformly in $(v,x,w) \in \mathcal{I} \times \mathcal{R}_{X,W}$, where
    \begin{align*}
         \omega_i(v,x,w) & = S_{V_T\mid X,W,Z}(v\mid x,w,Z_i)\zeta_i(v,x,w,Z_i) -p_{W\mid V_T,X,Z}(w\mid v,x,Z_i)S_{V_T}(v)\xi_{i}(v,x,Z_i)   \\ & \quad +\mathbbm{1}(X_i=x,W_i=w)\left[F_{V_T\mid X,W,Z}(v\mid x,w,Z_i)-F_{V_T\mid X,W}(v\mid x,w)\right].
    \end{align*}
\end{theorem}
Note that $F_{V_T\mid X,W}(v\mid x,w)=v$ for all $(v,x,w) \in \mathcal{I}\times \mathcal{R}_{X,W}$ under $H_0$. During the proof of this theorem, we also show in Appendix \ref{secapplemma} that $$\sup_{v \in \mathcal{I}}\left\lvert\widehat{F}_{\widehat{V}_T}(v)-v\right\rvert =o_p(n^{-1/2}),$$ meaning that $\widehat{F}_{\widehat{V}_T}(v)$ converges to $v$ at a rate faster than the usual parametric $n^{-1/2}$-rate. While this was already shown by \cite{feve2018estimation} for uncensored conditional ranks, it is not obvious that the result would still hold in the presence of right censoring. Further, let $\ell^{\infty}(\mathcal{I}\times \mathcal{R}_{X,W})$ be the set of all uniformly bounded real-valued functions equipped with the supremum norm. We are now ready to give the main result regarding the limiting distributions of the proposed test statistics.
\begin{theorem}\label{TH2}
    Under Assumptions \ref{Asupport} - \ref{ascoonvpi} and $H_0$, the empirical process $$\left\{n^{1/2}\left[\widehat{F}_{\widehat{V}_T\mid X,W}(v\mid x,w) - \widehat{F}_{\widehat{V}_T}(v)\right] : (v,x,w) \in \mathcal{I} \times  \mathcal{R}_{X,W}\right\},$$ converges weakly in $\ell^{\infty}(\mathcal{I}\times \mathcal{R}_{X,W})$ to a Gaussian process $\mathbb{G}(v,x,w)$ with mean zero and covariance function $$\text{Cov}\left(\frac{\omega_1(v_1,x_1,w_1)}{p_{X,W}(x_1,w_1)},\frac{\omega_1(v_2,x_2,w_2)}{p_{X,W}(x_2,w_2)}\right),$$ where $\omega_1(\cdot)$ is defined in Theorem \ref{TH1}. Moreover, we have that
    $$
    T_n^{KS} \xrightarrow{d} \sup_{v \in \mathcal{I},(x,w) \in \mathcal{R}_{X,W}}\left\lvert \mathbb{G}(v,x,w) \right\rvert,
    $$
    and
    $$
    T^{CM}_{n} \xrightarrow{d} \sum_{(x,w) \in \mathcal{R}_{X,W}} \pi(x,w)\int_{\mathcal{I} }  \mathbb{G}^2(v,x,w) \diff F_{V_T}(v),
    $$
    where $\pi(x,w)$ is defined by Assumption \ref{ascoonvpi}.
\end{theorem}
An explicit expression for the covariance function is not provided, as it would be excessively long and almost infeasible to estimate in practice. Instead, the limiting distributions are approximated using two possible bootstrap approaches that are detailed in Section \ref{secboot}. Further, consider the local alternative $$H_a :F_{V_T\mid X,W}(v\mid x,w)=v+n^{-1/2}H(v,x,w) \text{ for all } (v,x,w) \in \mathcal{I} \times \mathcal{R}_{X,W},$$ where for all $(x,w) \in \mathcal{R}_{X,W}$ the function $H$ is such that $F_{V_T\mid X,W}$ remains a valid conditional distribution function under $H_a$. It follows that, under $H_a$, the test statistics converge in distribution to the same limiting distribution as under $H_0$, except for the additive bias $H(v,x,w)$, that is
$$
    T_n^{KS} \xrightarrow{\textit{d}} \sup_{v \in \mathcal{I},(x,w) \in \mathcal{R}_{X,W}}\left\lvert \mathbb{G}(v,x,w) +H(v,x,w) \right\rvert,
    $$
    and
    $$
    T^{CM}_{n} \xrightarrow{\textit{d}} \sum_{(x,w) \in \mathcal{R}_{X,W}} \pi(x,w) \int_{\mathcal{I} }  \left[\mathbb{G}(v,x,w)+H(v,x,w)\right]^2 \diff F_{V_T}(v).
    $$
Therefore, the proposed tests can detect local alternatives that converge to $H_0$ at the parametric $n^{-1/2}$-rate.
\subsection{Bootstrap approximations}\label{secboot}

Due to the complicated covariance structure of the process $\mathbb{G}(v,x,w)$, we use bootstrap approximations for the limiting distributions of $T_n^{KS}$ and $T_n^{CM}$. For the bootstrap procedure to be valid, we follow the approach of \cite{feve2018estimation} and impose the slightly stronger null hypothesis $$H^*_0 : V_T \indep (X,W,Z),$$ instead of $H_0:V_T \indep (X,W)$. Note that $H^*_0$ is only slightly stronger than $H_0$ since $V_T \indep (X,Z)$ by construction. We propose two bootstrap procedures that only differ in the construction of the bootstrap censoring times. Procedure A is the standard bootstrap for right censored data \citep{efron1981censored}, in which both the duration and censoring times are generated from (conditional) Kaplan-Meier fits. Procedure B is based on the conditional random censoring algorithm \citep{robinson1983bootstrap}. In this approach, the known censoring times from the original sample are kept for the censored observations. For the uncensored observations, censoring times are generated from a conditional Kaplan-Meier fit, conditionally on the information that the unknown censoring time is greater than the observed duration time. In Section \ref{secfinitesample}, we investigate whether there are scenarios under which one of the bootstrapping approaches performs better. The two procedures (Type A and B) are implemented as follows.

\bigskip \noindent For each $b \in \{1,\dots,B\}:$
\begin{enumerate}
    \item[1.] For $i = 1,\dots,n$, let $X^*_i=X_i$, $W^*_i=W_i$ and $Z^*_i=Z_i$.
    \item[2.]  Let $U^*_{T,1},\dots,U^*_{T,n}$ be i.i.d. variables randomly drawn from $\mathcal{U}[0,1]$ and $\widehat{F}^{-1}_{C \mid X,Z}(u \mid X_i,Z_i) = \inf\{C_i:\widehat{F}_{C \mid X,Z}(C_i \mid X_i,Z_i) \geq u \}$, where $\widehat{F}_{C \mid X,Z}$ is calculated with a conditional Kaplan-Meier estimator that uses $(1-\Delta)$ as the censoring indicator.
       \begin{enumerate}
       \item[A.]  Let $U^*_{C,1},\dots,U^*_{C,n}$ be i.i.d. variables randomly drawn from $\mathcal{U}[0,1]$. For $i = 1,\dots,n$, let $C^*_i = \widehat{F}^{-1}_{C \mid X,Z}(U^*_{C,i} \mid X_i,Z_i)$.
    \item[B.]  Let $U^*_{C,1},\dots,U^*_{C,n}$ be i.i.d. variables randomly drawn from $\mathcal{U}[\widehat{F}_{C \mid X,Z}(Y_i \mid X_i,Z_i),1]$. For $i = 1,\dots,n$, let $C^*_i = \Delta_i \times \widehat{F}^{-1}_{C \mid X,Z}(U^*_{C,i} \mid X_i,Z_i) + (1-\Delta_i) \times C_i$.
    \end{enumerate}
    \item[3.] Next, let $Y^*_i = \min\left\{\widehat{F}^{-1}_{T \mid X,Z}(U^*_{T,i}\mid X_i,Z_i),C^*_i\right\}$ and $\Delta^*_i=\mathbbm{1}\left\{Y^*_i = \widehat{F}^{-1}_{T \mid X,Z}(U^*_{T,i}\mid X_i,Z_i)\right\}$.
    \item[4.] Using  $\left\{Y^*_i,\Delta^*_i,X_i,W_i,Z_i\right\}_{i=1}^n$ calculate the test statistic $T^{*}_{n,b}$.
\end{enumerate}
Finally, compute the $p-$value as $B^{-1}\sum^B_{b=1} \mathbbm{1}\left\{T^*_{n,b} > T_{n}\right\}$.

\section{Finite sample study}\label{secfinitesample}

In this section, we investigate the finite-sample performance of the proposed tests described in Section \ref{sectests} and both of the bootstrap approximations described in Section \ref{secboot} through Monte Carlo simulations. To estimate the Cramér-von Mises test statistic, we set $\widehat{\pi}(x,w) = 1$ for all simulation settings. Moreover, we implement the warp-speed method of \cite{warpspeed} to obtain the critical values. The procedure calculates a single bootstrap test statistic for each Monte Carlo sample and aggregates these statistics across simulations to approximate the critical value, which significantly reduces computation time. We will begin by describing the data-generating process, followed by a discussion of the proposed tests' performance under different degrees of endogeneity, censoring and instrument strength.

\subsection{Data-generating process}\label{secDGP}

For $i=1,\dots,n$, let $$U_{T,i} \sim \mathcal{U}[0,1], \text{ } U_{C,i} \sim \mathcal{U}[0,1]\text{ and } X_i\sim \text{Bin}(0.45), $$ 
such that $$\pi_{W,i}=\frac{\exp(0.9 - 0.3 X_i)}{1+\exp(0.9 - 0.3 X_i)},$$  and $W_i \sim \text{Bin}(\pi_{W,i})$. Moreover, let $$\pi_{Z,i}=\frac{\exp(-2+0.2X_i+\eta W_i+\alpha \bar{U}_{T,i})}{1+\exp(-2+0.2X_i+\eta W_i+\alpha \bar{U}_{T,i})},$$ with $\bar{U}_{T,i} = U_{T,i}-0.5$ and $Z_i \sim \text{Bin}(\pi_{Z,i})$. The reason for centering $U_{T,i}$ is to keep the censoring rate somewhat stable when we vary $\alpha$, which determines the endogeneity of $Z$. This will allow us to compare the power of the test statistics under different degrees of endogeneity. Moreover, $Z$ is exogenous when $\alpha=0$ and $\eta$ controls the instrument strength. 
Further, let 
$$
T_i = \exp\left\{4-0.5X_i-Z_i +\Phi^{-1}(U_{T,i})\right\},
$$
with $\Phi^{-1}$ the inverse of the standard normal distribution. Furthermore, let $$C_i=-\frac{\log(1-U_{C,i})}{\exp\left\{\lambda+0.9X_i+0.8Z_i\right\}}.$$ Note that $\lambda$ controls the censoring rate. Lastly, we generate
$Y_i=\min(T_i,C_i)$ and $\Delta_i=\mathbbm{1}(Y_i=T_i)$ to get the simulated data $\left\{Y_i,\Delta_i,X_i,W_i,Z_i\right\}_{i=1,\dots,n}$. Even though we only have one covariate $X$, the parameter values were chosen such that there are very few data points for which $W_i=0$ and $Z_i=1$. On average, 1.9\% of the simulated data have $(X_i,W_i,Z_i)=(0,0,1)$ and 2.3\% have $(X_i,W_i,Z_i)=(1,0,1)$. This is a similar situation to the empirical application described in Section \ref{secempapp}. For each of the following simulation settings, we performed 1000 Monte Carlo replications.

\subsection{Endogeneity and sample size}

To generate the data used in this part of the simulation study, we let $(\eta,\lambda)=(2.4,-5.7)$. By choosing these parameter values, we have an instrument strength of around $\tau_{WZ}=0.45$ (measured by Kendall's Tau between $W$ and $Z$) and approximately 25\% censoring. We use Kendall's tau as a proxy for instrument strength because it is simple to interpret and report. Note that these values were chosen such that the instrument strength and censoring rate are comparable to the empirical strata examined in Section \ref{secempapp}.

Looking at Figure \ref{distributiontests}, we see that both the Kolmogorov-Smirnov (KS) and Cramér-von Mises (CM) test statistics are able to discriminate the null from the alternative hypotheses. It is also clear that the CM test statistic performs better, where for a sample size of $n=1000$, the distributions for $\alpha = 0$ and $\alpha =5$ are almost completely separated. Figure \ref{plotpower} shows the power of the two test statistics as a function of the degree of endogeneity $\tau_{ZU}$ (measured by Kendall's Tau between $Z$ and $U_T$) for different sample sizes and both of the proposed bootstrap approximations. Again, we see that the CM test statistic outperforms the KS test statistic. As the endogeneity and/or the sample size increase, so does the power of the test statistics. Under the null $(\tau_{ZU} = 0)$, both test statistics reject at the nominal level of 5\%, independent of the sample size. Comparing the two bootstrap approximations, there seems to be no noticeable difference. Looking at Figure \ref{pvalue500}, we see that the $p$-values obtained from Monte Carlo are well approximated by both bootstrapping procedures. However, approach B seems to be slightly more conservative. Overall, we find that our test statistics have good power, even when there are almost no observations with $W=0$ and $Z=1$. Moreover, the CM test statistic consistently outperforms the KS test statistic regardless of sample size or the degree of endogeneity.
\begin{figure}[h!]
\centering
    \includegraphics[width=\linewidth]{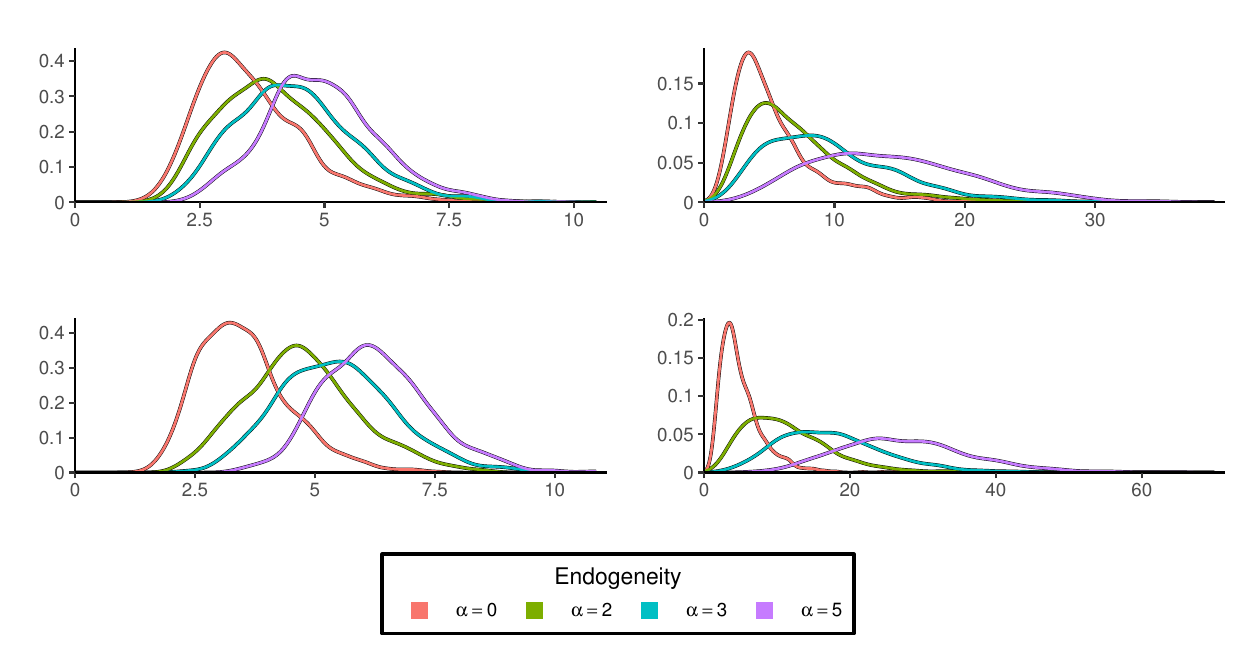}
    \caption{The distribution of $T^{KS}_n$ (left) and $T^{CM}_n$ (right) under the null and various alternative hypotheses (increasing degree of endogeneity $\alpha$). The first row shows the results for $n=500$ and the second row for $n=1000$.}
    \label{distributiontests}
\end{figure}
     \begin{figure}[h!]
\centering
    \includegraphics[width=\linewidth]{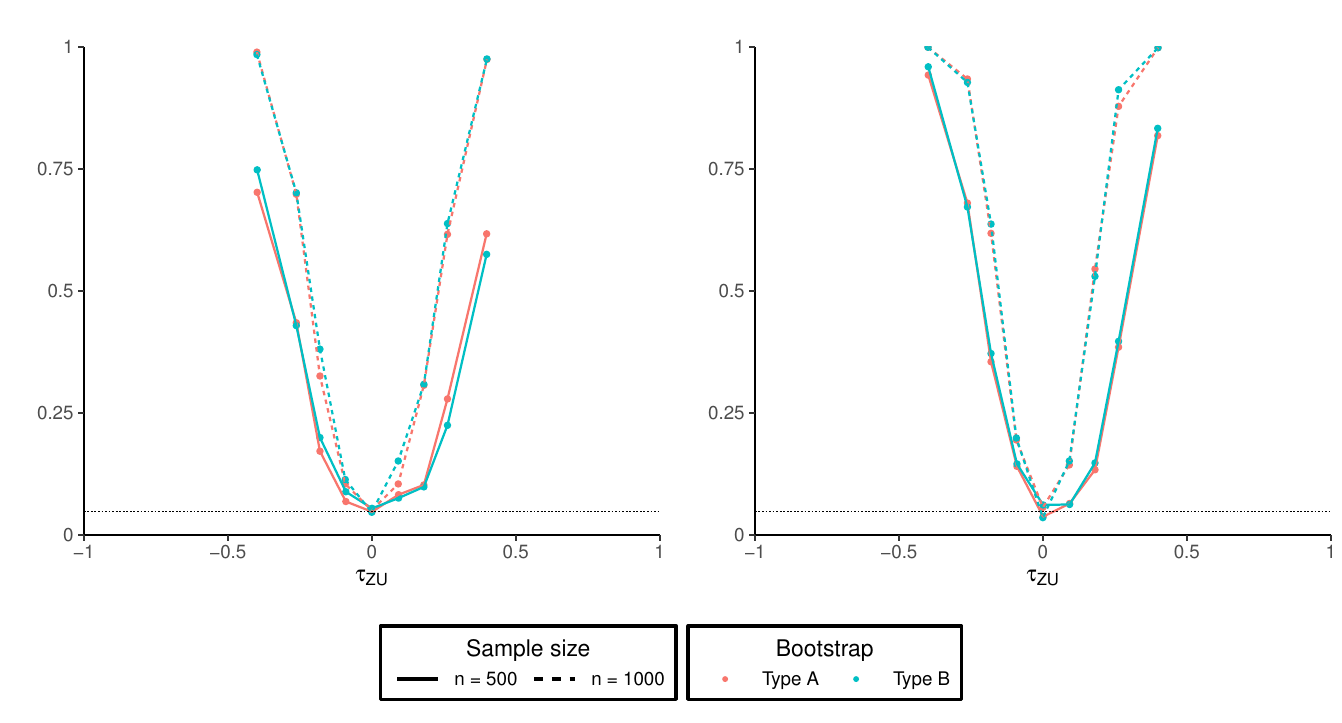}
    \captionof{figure}{The power of $T^{KS}_n$ (left) and $T^{CM}_n$ (right) with respect to the degree of endogeneity $\tau_{ZU}$ for $n=500$, $n=1000$ and both bootstrap approximations (Type A and B). }
    \label{plotpower}
    \end{figure}
        \begin{figure}[h!]
\centering
    \includegraphics[width=\linewidth]{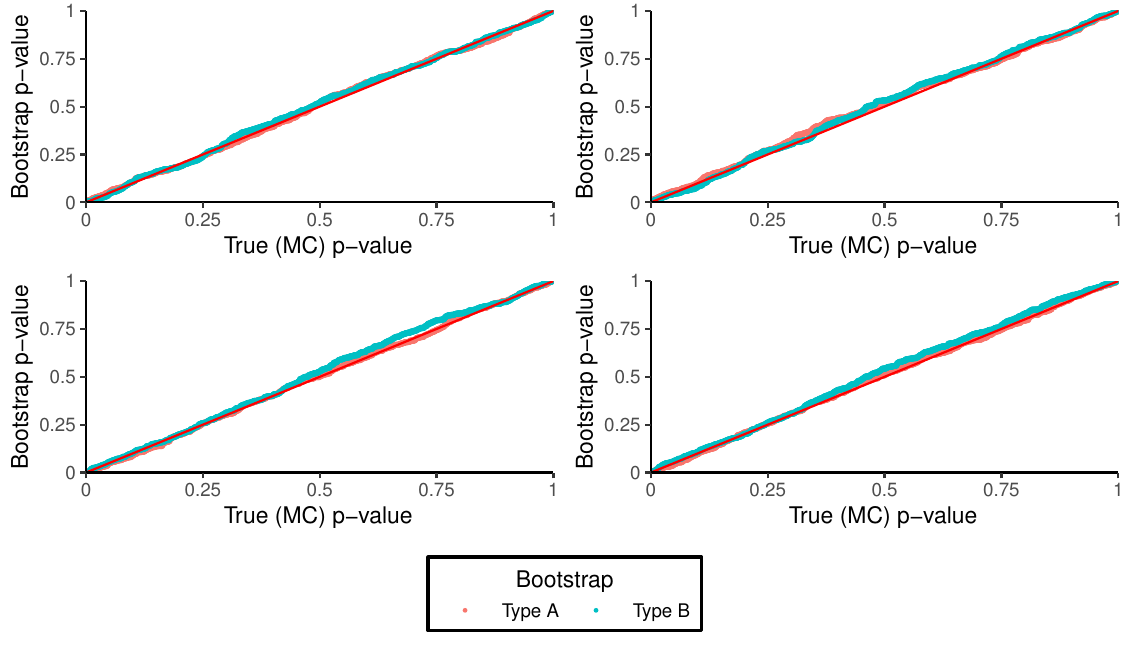}
    \captionof{figure}{The Monte Carlo and bootstrap $p$-values for $T^{KS}_n$ (left) and $T^{CM}_n$ (right) under the null for both bootstrap approximations (Type A and B). The first row shows the results for $n=500$ and the second row for $n=1000$. }
    \label{pvalue500}
    \end{figure}
\subsection{Instrument strength and censoring rate}
In this subsection, we will separately vary $\eta$ and $\lambda$ alongside $\alpha$ to assess the influence of the instrument strength and censoring rate, respectively, on the power of our proposed test statistics. All simulations use a sample size of $n=1000$ and the other parameters have the same value as described in Section \ref{secDGP}. We start by varying the instrument strength $\eta \in \{0.4,0.9,1.4,1.9,2.4,2.9,3.4\}$ and fixing $\lambda =-5.7$, such that there is around 25\% censoring. Figure \ref{plotIV} shows us that the power increases when the instrument strength increases, as would be expected. Moreover, when there is weak endogeneity ($\alpha =2$), the increase in power becomes less as the instrument strength increases. The CM test statistic again consistently outperforms the KS test statistic. Concerning the bootstrap approximations, it can be seen that approach B performs slightly worse than approach A for the KM test statistic when there is weak endogeneity. Under the null, the rejection rate remains at the nominal 5\% level for both approximations, independent of the instrument strength. 

     \begin{figure}[h!]
\centering
    \includegraphics[width=\linewidth]{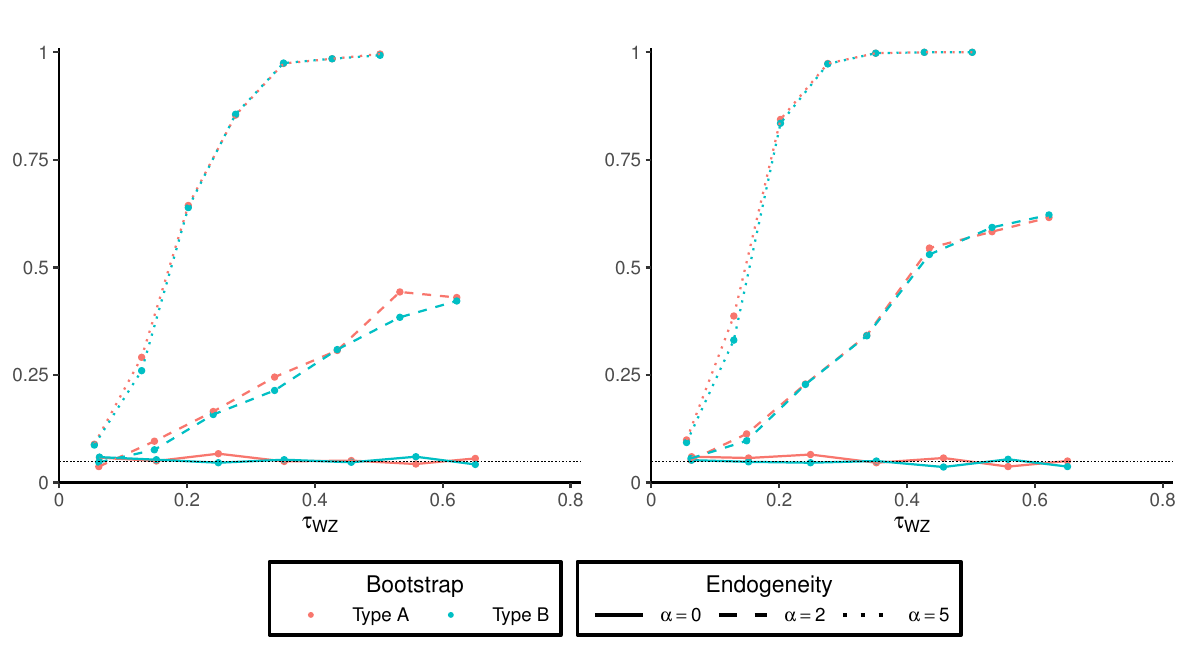}
    \captionof{figure}{The power of $T^{KS}_{1000}$ (left) and $T^{CM}_{1000}$ (right) with respect to the instrument strength $\tau_{WZ}$ for both bootstrap approximations (Type A and B) and increasing degree of endogeneity $(\alpha)$. }
    \label{plotIV}
    \end{figure}

\clearpage To examine the impact of censoring, we vary the parameter $\lambda \in \{-5.7,-5.1,-4.6,-4.2,-3.8\}$ and fix $\eta=2.4$ such that there is an instrument strength of around $\tau_{WZ}=0.45$. As we would expect, Figure \ref{plotcens} indicates that the power of both test statistics decreases when there is more censoring. However, until approximately 50\% censoring, the power of the CM test statistic remains relatively stable. When there is weak endogeneity, we see that bootstrap approach B performs slightly better when the censoring rate is above 50\%. Interestingly, the CM test statistic performs a little worse than the KM test statistic when there is weak endogeneity and high censoring. The rejection rate under the null remains somewhat stable as the censoring rate increases, but falls slightly below 5\% when the censoring rate is around 75\%. Overall, the test statistics still perform well despite reduced instrument strength and high censoring rates.

     \begin{figure}[ht]
\centering
    \includegraphics[width=\linewidth]{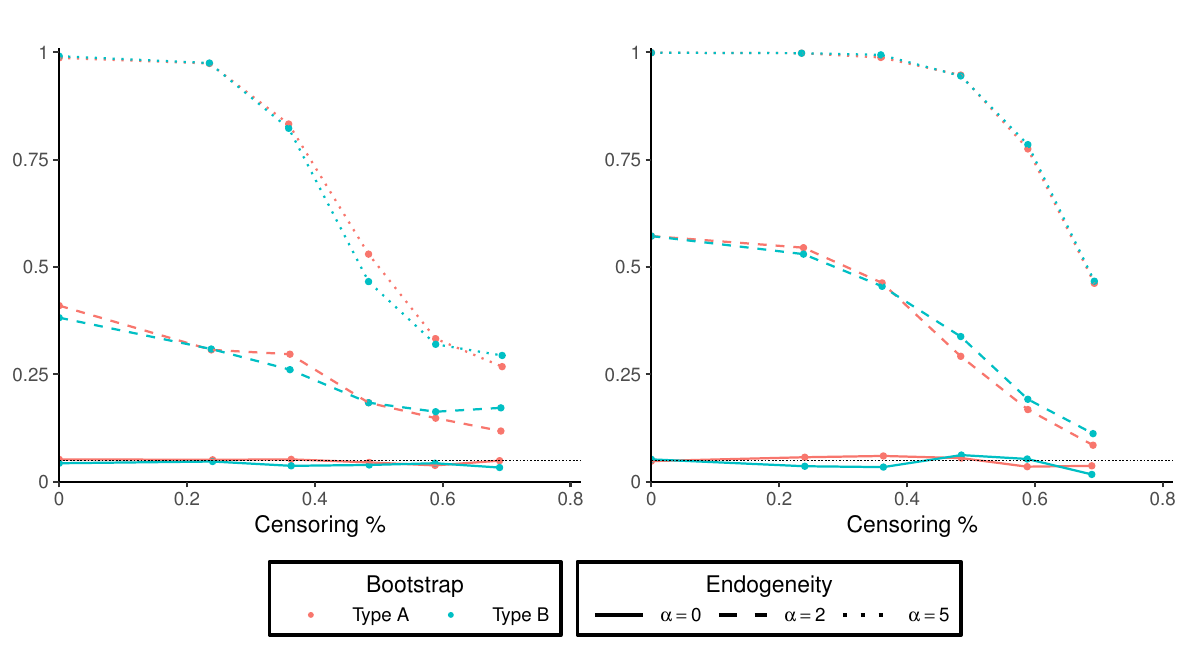}
    \captionof{figure}{The power of $T^{KS}_{1000}$ (left) and $T^{CM}_{1000}$ (right) with respect to the censoring rate for both bootstrap approximations (Type A and B) and increasing degree of endogeneity $(\alpha)$. }
    \label{plotcens}
    \end{figure}

\section{Empirical application}\label{secempapp}

The data come from the National Job Training Partnership Act (JTPA) Study, a large randomized evaluation of more than 600 federally funded programs designed to increase the employability of eligible adults and out-of-school youths. The study was conducted between 1987 and 1989, enrolled over 20000 applicants and collected follow-up information on employment, earnings, and program participation. Random assignment placed unemployed individuals into either a treatment group (eligible for JTPA services) or a control group (ineligible for 18 months). Nevertheless, roughly 3\% of control-group members received JTPA services despite their ineligibility. Because individuals may self-select into treatment in a nonrandom way, noncompliance can induce endogeneity of the treatment \citep{angrist1996identification}. All participants were surveyed between 12 and 36 months after randomization (average of 21 months). A second follow-up survey was administered to a subsample of 5,468 respondents, focusing on the interval between the two interviews, and took place between 23 and 48 months after the initial randomization. The data can be downloaded at \url{https://www.upjohn.org/data-tools/employment-research-data-center/national-jtpa-study}.

Because the study combines a clear policy intervention with extensive follow-up on earnings and labor-market outcomes, the data have become a standard for evaluating causal estimands in both cross-sectional and duration contexts. In particular, \cite{bloom1997benefits}, \cite{abadie2002instrumental} and \cite{wuthrich2020comparison} investigate the impact of JTPA services on the sum of earnings after treatment. The effect of JTPA services on unemployment duration has been investigated by \cite{frandsen2015treatment} and \cite{beyhum2024instrumental} under the (conditional) independent censoring assumption, while \cite{crommen2024instrumental} and \cite{crommen2025estimation} allow for dependent censoring. To make the independent censoring assumption more plausible, \cite{frandsen2015treatment} and \cite{beyhum2024instrumental} rely exclusively on data from the first follow-up survey. Because almost all participants were surveyed, it is reasonable to assume that the censoring is administrative. By contrast, \cite{crommen2024instrumental} and \cite{crommen2025estimation} incorporate data from the second follow-up survey. They argue that participants who were invited to a second survey but did not participate could introduce endogenous censoring. Because the validity of our test statistics relies on the conditional independent censoring assumption, we restrict our empirical analysis to data from the first follow-up interview. Our goal is to test whether instrumental variable methods are required to identify the causal effect of JTPA services on unemployment duration, i.e., if JTPA services are an endogenous treatment.
    
Due to the size of the data set, we will focus our attention on two strata. The first stratum consists of 1127 single white men without children who are 30 years or younger and reported having no job at the time of treatment assignment. The other stratum is similar to the first one, except that the 1017 men are non-white. We will refer to these strata as white and non-white men, respectively. The covariate \textit{HSGED} equals 1 if an individual held a high school diploma or GED at the time of treatment assignment and 0 otherwise. Approximately 46\% of white men and 40\% of non-white men had a high school diploma or obtained a GED at the time of treatment assignment. Moreover, treatment assignment will be our instrument $W$ ($W=1$ if assigned to the treatment group, $0$ for the control group). We deem treatment assignment to be a valid instrument as it is randomly assigned, correlated with JTPA participation (Kendall's Tau of 0.457 for white men and 0.475 for non-white men) and influences time to employment only through treatment participation. Note that $Z=1$ if they participated in JTPA services, and $0$ otherwise. Around 68\% of white men and 71\% of non-white men were assigned to the treatment group. In both strata, only 47\% actually participated in JTPA services. 
    
Figure \ref{plotcenshist} displays histograms of the observed follow-up times for each stratum, where a darker shading indicates a higher censoring rate. Consistent with \cite{frandsen2015treatment}, we find that most observations are censored at approximately 600 days, around which time most of the follow-up interviews took place. Beyond 600 days, nearly all observations are censored. Moreover, Table \ref{tablesum} shows the censoring rates for each level of \textit{HSGED}, $W$ and $Z$. keeping \textit{HSGED} and $Z$ fixed, we find that the censoring rate remains stable for different levels of $W$. Therefore, we do not find any major violations of Assumption \ref{AindepTC}. Table \ref{tablesum} also shows the number of observations for each level of \textit{HSGED}, $W$ and $Z$. As expected, there are very few observations from the control group who participated in JTPA services, i.e., $W=0$ and $Z=1$. However, as shown in Section \ref{secfinitesample}, the test still performs well when there are only a few observations with $W=0$ and $Z=1$. 

\begin{figure}[ht]
\centering
    \includegraphics[width=0.9\linewidth]{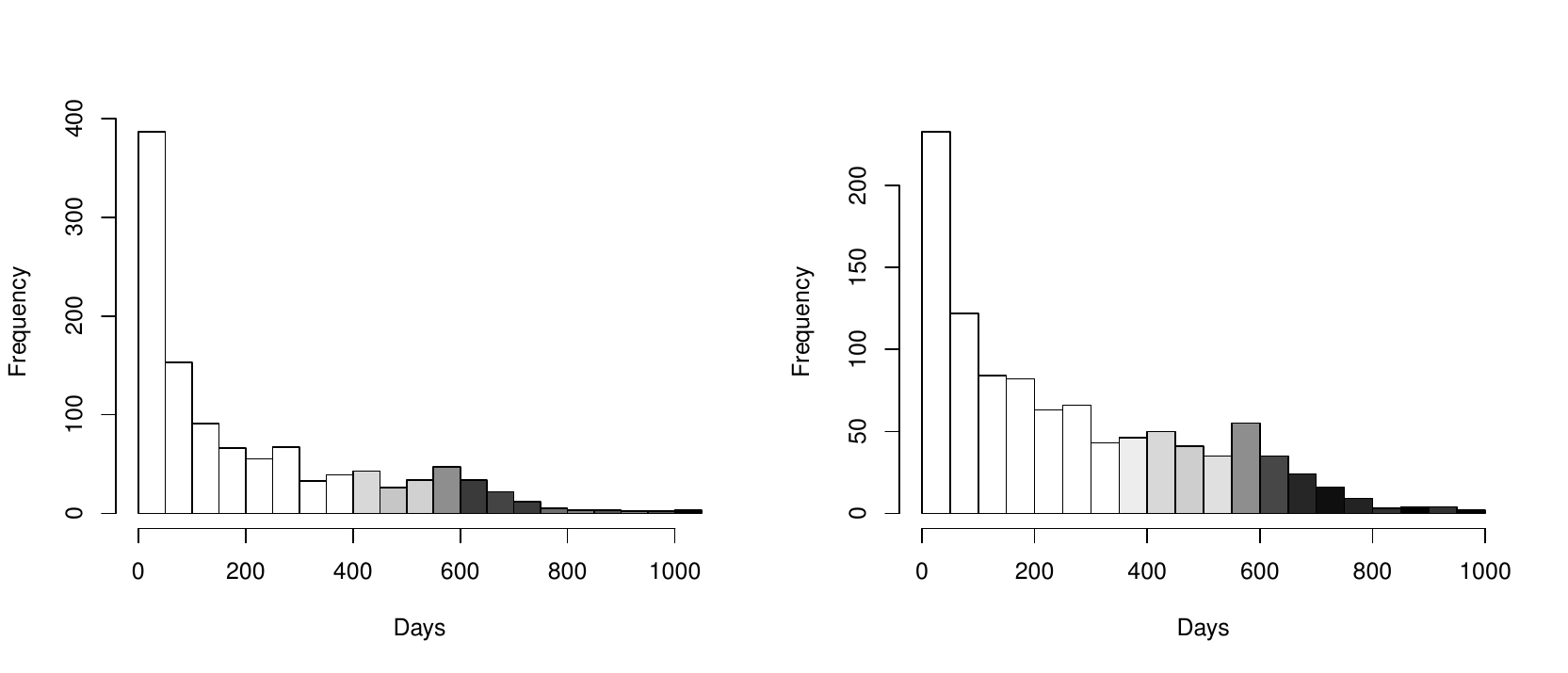}
    \caption{Histograms of the observed follow-up times $Y$ for white men (left) and non-white men (right), in days, starting from random assignment. The darker the shade, the higher the censoring rate. }
    \label{plotcenshist}
    \end{figure}

\begin{table}[ht]
  \centering
  \setlength{\tabcolsep}{7pt}
  \renewcommand{\arraystretch}{1.3}

  \begin{subtable}{0.48\textwidth}
    \centering
    \begin{tabular}{|c|c|c|}
      \hline
      $W=0$ & $Z=0$ & $Z=1$ \\ \hline
      $\mathrm{HSGED}=0$ & $\displaystyle\begin{array}{c}(171,89.5\%)\end{array}$ & $\displaystyle\begin{array}{c}(32,90.6\%)\end{array}$ \\ \hline
      $\mathrm{HSGED}=1$ & $\displaystyle\begin{array}{c}(141,90.8\%)\end{array}$ & $\displaystyle\begin{array}{c}(22,90.9\%)\end{array}$ \\ \hline
    \end{tabular}
  \end{subtable}\hfill
  \begin{subtable}{0.48\textwidth}
    \centering
    \begin{tabular}{|c|c|c|}
 \hline
      $W=1$ & $Z=0$ & $Z=1$ \\ \hline
      $\mathrm{HSGED}=0$ & $\displaystyle\begin{array}{c}(152,87.5\%)\end{array}$ & $\displaystyle\begin{array}{c}(255,90.6\%)\end{array}$ \\ \hline
      $\mathrm{HSGED}=1$ & $\displaystyle\begin{array}{c}(126,87.3\%)\end{array}$ & $\displaystyle\begin{array}{c}(228,89.9\%)\end{array}$ \\ \hline
    \end{tabular}
  \end{subtable}
  \par\vspace{12pt}
  \begin{subtable}{0.48\textwidth}
    \centering
    \begin{tabular}{|c|c|c|}
          \hline
      $W=0$ & $Z=0$ & $Z=1$ \\ \hline
      $\mathrm{HSGED}=0$ & $\displaystyle\begin{array}{c}(154,79.2\%)\end{array}$ & $\displaystyle\begin{array}{c}(12,91.7\%)\end{array}$ \\ \hline
      $\mathrm{HSGED}=1$ & $\displaystyle\begin{array}{c}(112,86.6\%)\end{array}$ & $\displaystyle\begin{array}{c}(19,84.2\%)\end{array}$ \\ \hline
    \end{tabular}
  \end{subtable}\hfill
  \begin{subtable}{0.48\textwidth}
    \centering
    \begin{tabular}{|c|c|c|}
      \hline
      $W=1$ & $Z=0$ & $Z=1$ \\ \hline
      $\mathrm{HSGED}=0$ & $\displaystyle\begin{array}{c}(173,79.2\%)\end{array}$ & $\displaystyle\begin{array}{c}(272,90.1\%)\end{array}$ \\ \hline
      $\mathrm{HSGED}=1$ & $\displaystyle\begin{array}{c}(96,87.5\%)\end{array}$ & $\displaystyle\begin{array}{c}(179,87.7\%)\end{array}$ \\ \hline
    \end{tabular}
  \end{subtable}
  \caption{Number of observations and censoring rate, respectively, for each level of high school diploma/GED status (\textit{HSGED}), JTPA participation $(Z)$ and treatment assignment ($W$) for white men (top row) and non-white men (bottom row).}
  \label{tablesum}
\end{table}

We compute the $p$-values of both test statistics using the two bootstrap approximations described in Section \ref{secboot}, based on 1000 resamples. Similarly to Section \ref{secfinitesample}, we set $\widehat{\pi}(x,w)=1$ for the Cramér–von Mises test statistic. The results are reported in Table \ref{tablepvalue} and Figure \ref{KMstrata} displays the Kaplan–Meier curves, conditional on high school diploma/GED status and treatment participation, for both the white and non-white men. For the stratum of white men, we fail to reject the null hypothesis that $Z$ is exogenous. Therefore, we can simply compare the estimated Kaplan-Meier curves for the treated and untreated. A log-rank test indicates no statistically significant difference between the survival curves ($p$-value of 0.330). Even after conditioning on having a high school diploma or GED, no significant difference is found ($p$-value of 0.503 for both levels of \textit{HSGED}). We therefore conclude that, within this stratum, there is no evidence of a significant effect of JTPA services on unemployment duration. 

\begin{table}[ht]
  \centering
  \setlength{\tabcolsep}{7pt}
  \renewcommand{\arraystretch}{1.3}

  \begin{subtable}{0.48\textwidth}
    \centering
    \begin{tabular}{|c|c|c|}
      \hline
      white men &  Type A & Type B \\ \hline
      $T^{KS}_{1127}$ & $\displaystyle\begin{array}{c}0.595\end{array}$ & $\displaystyle\begin{array}{c}0.572\end{array}$ \\ \hline
      $T^{CM}_{1127}$ & $\displaystyle\begin{array}{c}0.399\end{array}$ & $\displaystyle\begin{array}{c}0.392\end{array}$ \\ \hline
    \end{tabular}
  \end{subtable}\hfill
  \begin{subtable}{0.48\textwidth}
    \centering
    \begin{tabular}{|c|c|c|}
      \hline
         non-white men &  Type A & Type B \\ \hline
      $T^{KS}_{1017}$ & $\displaystyle\begin{array}{c}0.003\end{array}$ & $\displaystyle\begin{array}{c}0.002\end{array}$ \\ \hline
      $T^{CM}_{1017}$ & $\displaystyle\begin{array}{c}0.030\end{array}$ & $\displaystyle\begin{array}{c}0.017\end{array}$ \\ \hline
    \end{tabular}
  \end{subtable}
    \caption{Bootstrap $p$-values of the Kolmogorov–Smirnov (KS) and Cramér–von Mises (CM) test statistics for both bootstrap approximations (Type A and B), for white men (left) and non-white men (right).}
  \label{tablepvalue}
\end{table}

   \begin{figure}[ht]
\centering
    \includegraphics[width=\linewidth]{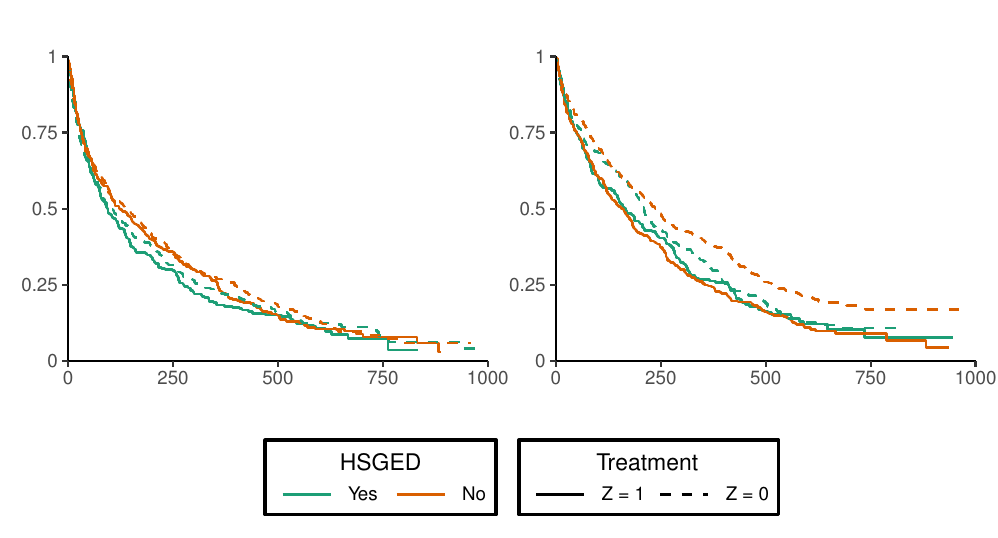}
    \captionof{figure}{Kaplan-Meier curves, conditional on high school diploma/GED status (\textit{HSGED}) and treatment participation ($Z$) for white men (left) and non-white men (right). }
  \label{KMstrata}
\end{figure}

By contrast, for the stratum of non-white men, the null hypothesis that $Z$ is exogenous is rejected. This suggests the presence of unobserved heterogeneity influencing both treatment participation and unemployment duration. A log-rank test further reveals a significant difference between the survival functions for the treated and untreated ($p$-value = 0.0003). We find that, for non-white men without a high school diploma or GED, the treated have a significantly lower probability of being unemployed compared to the untreated ($p$-value of $8.87\times 10^{-5}$). However, for non-white men with a high school diploma or GED, there seems to be no significant difference ($p$-value of 0.403). It is important to emphasize that because $Z$ is endogenous in this stratum, we cannot interpret the observed survival differences as causal effects of JTPA participation. The effect of JTPA services among non-white men without a high school diploma or GED may, for example, reflect self-selection of more motivated or higher-ability individuals into treatment. To actually estimate the causal effect of JTPA participation under endogeneity, one would need to use instrumental variable methods that allow for right censoring. Possible approaches include the nonparametric local average treatment effect framework of \cite{frandsen2015treatment}, as well as the population-level treatment effect methods developed by \cite{beyhum2022nonparametric,beyhum2024instrumental}, which are nonparametric and semiparametric, respectively.

\section{Conclusion and future research}\label{secconclusion}

This paper develops nonparametric tests for exogeneity in nonparametric nonseparable duration models subject to right censoring. The tests are based on the independence between the conditional ranks, which are also right censored, and the instrumental variable. The proposed tests do not require estimation of the underlying structural model using nonparametric instrumental variable methods, which is even more challenging in the presence of right censoring. We derive the asymptotic properties of the proposed tests and illustrate their power under various scenarios using Monte Carlo simulations. The validity of the two bootstrap approaches to approximate the critical values is also analyzed via simulations. We find that the Cramér–von Mises statistic consistently outperforms the Kolmogorov–Smirnov statistic, and that both bootstrapping approaches perform equally well in almost all scenarios. An empirical application to the National Job Training Partnership Act (JTPA) Study illustrates the usefulness of the tests in applied settings. In particular, it provides practitioners with a test to decide if standard survival comparisons are appropriate or whether instrumental variable methods that allow for right censoring are needed to correctly identify treatment effects.

The tests can be extended in multiple directions. Firstly, extending the methodology to accommodate continuous covariates, instruments and treatments would substantially broaden its applicability. The most straightforward way to achieve this is by replacing the conditional Kaplan–Meier estimator with a smooth conditional estimator, such as \cite{beran1981nonparametric}'s estimator. Although conceptually straightforward, this extension introduces both theoretical and practical complications, such as bandwidth selection and the curse of dimensionality. A way of dealing with this could be to instead use semiparametric models, such as the \cite{CoxPHmodel} proportional hazards model. Another important extension is to relax the conditional independent censoring assumption and allow for some forms of dependence between $T$ and $C$ after conditioning on $(X,W,Z)$. Copula-based methods, such as the copula-graphic estimator \citep{Rivest2001AMA,ZHENGMING1995Eoms}, offer a natural approach to modeling such dependencies. However, other copula-based methods for dependent censoring could also be considered (see \cite{crommen2025recent} for a recent review).

\section*{Acknowledgments and funding}
The computational resources and services used in this work were provided by the VSC (Flemish Supercomputer Center), funded by the Research Foundation Flanders (FWO) and the Flemish Government department EWI. G. Crommen is funded by a PhD fellowship from the Research Foundation - Flanders (grant number 11PKA24N). J.P. Florens acknowledges funding from the French National Research Agency (ANR) under the Investments for the Future (Investissement d'Avenir), grant ANR-17-EURE-0010. I. Van Keilegom acknowledges funding from the FWO and F.R.S. - FNRS (Excellence of Science programme, project ASTeRISK, grant no. 40007517), and from the FWO (senior research projects fundamental research, grant no. G047524N). Moreover, the authors would like to thank Jad Beyhum for helpful discussions.

\bibliography{literature}

\clearpage 

\appendix

\appendixpage

Throughout the Appendix, let 

	\begin{alignat*}{2}
		&\widehat{S}_{\widehat{V}}(v)=n^{-1}\sum^n_{i=1}\mathbbm{1}(\widehat{F}_{T \mid X,Z}(Y_i \mid X_i,Z_i) > v),  && \quad S_{\widehat{V}}(v) = \mathbbm{P}(\widehat{F}_{T \mid X,Z}(Y \mid X,Z) > v),  \\ &  \widehat{S}_{V}(v)=n^{-1}\sum^n_{i=1}\mathbbm{1}(F_{T \mid X,Z}(Y_i \mid X_i,Z_i) > v), && \quad  S_V(v)=\mathbbm{P}(F_{T \mid X,Z}(Y \mid X,Z)> v), \\ &  \widehat{S}_{\widehat{V},1}(v)=n^{-1}\sum^n_{i=1}\mathbbm{1}(\widehat{F}_{T \mid X,Z}(Y_i \mid X_i,Z_i) > v,\Delta_i =1), && \quad S_{\widehat{V},1}(v) = \mathbbm{P}(\widehat{F}_{T \mid X,Z}(Y \mid X,Z) > v, \Delta = 1),  \\ & \widehat{S}_{V,1}(v)=n^{-1}\sum^n_{i=1}\mathbbm{1}(F_{T \mid X,Z}(Y_i \mid X_i,Z_i) > v,\Delta_i =1), && \quad S_{V,1}(v)=\mathbbm{P}(F_{T \mid X,Z}(Y \mid X,Z)>v,\Delta=1),
	\end{alignat*}
where $S_{\widehat{V}}(v) = \mathbbm{P}(\widehat{F}_{T \mid X,Z}(Y \mid X,Z) > v)$ is calculated with respect to the law of $(Y,Z,X)$ conditional on $\widehat{F}_{T \mid X,Z}$ with $(Y,X,Z)$ independent of the data $\{Y_i,X_i,Z_i\}^n_{i=1}$ and  $S_{\widehat{V},1}(v) = \mathbbm{P}(\widehat{F}_{T \mid X,Z}(Y \mid X,Z) > v, \Delta = 1)$ is calculated with respect to the law of $(Y,\Delta,X,Z)$ conditional on $\widehat{F}_{T \mid X,Z}$ with $(Y,\Delta,X,Z)$ independent of the data $\{Y_i,\Delta_i,X_i,Z_i\}^n_{i=1}$. Moreover, let
\begin{alignat*}{2}
		&\widehat{S}_{\widehat{V},1\mid X,W,Z}(v \mid x,w,z)=\frac{\widehat{S}_{\widehat{V},X,W,Z,1}(v,x,w,z)}{\widehat{p}_{X,W,Z}(x,w,z)},  && \quad S_{\widehat{V},1\mid X,W,Z}(v \mid x,w,z)=\frac{S_{\widehat{V},X,W,Z,1}(v,x,w,z)}{p_{X,W,Z}(x,w,z)},  \\ &  \widehat{S}_{V,1\mid X,W,Z}(v \mid x,w,z)=\frac{\widehat{S}_{V,X,W,Z,1}(v,x,w,z)}{\widehat{p}_{X,W,Z}(x,w,z)}, && \quad  S_{V,1\mid X,W,Z}(v \mid x,w,z)=\frac{S_{V,X,W,Z,1}(v,x,w,z)}{p_{X,W,Z}(x,w,z)},
	\end{alignat*}
with
\begin{align*}
    & \widehat{S}_{\widehat{V},X,W,Z,1}(v,x,w,z)=n^{-1}\sum^n_{i=1} \mathbbm{1}(\widehat{F}_{T \mid X,Z}(Y_i \mid X_i,Z_i) > v,X_i=x, W_i = w, Z_i = z,\Delta_i=1),\\ & S_{\widehat{V},X,W,Z,1}(v,x,w,z)=\mathbbm{P}(\widehat{F}_{T \mid X,Z}(Y \mid X,Z) > v,X=x, W = w, Z = z,\Delta=1), \\ & \widehat{S}_{V,X,W,Z,1}(v,x,w,z)=n^{-1}\sum^n_{i=1} \mathbbm{1}(F_{T \mid X,Z}(Y_i \mid X_i,Z_i) > v,X_i=x, W_i = w, Z_i = z,\Delta_i=1),  \\ & S_{V,X,W,Z,1}(v,x,w,z)=\mathbbm{P}(F_{T \mid X,Z}(Y \mid X,Z) > v,X=x, W = w, Z = z,\Delta=1), \\ & \widehat{p}_{X,W,Z}(x,w,z)=n^{-1}\sum^n_{i=1}\mathbbm{1}(X_i=x,W_i=w,Z_i=z), \\ &   p_{X,W,Z}(x,w,z)=\mathbbm{P}(X=x,W=w,Z=z),
\end{align*}
where $S_{\widehat{V},X,W,Z,1}(v,x,w,z) = \mathbbm{P}(\widehat{F}_{T \mid X,Z}(Y \mid X,Z) > v, X=x,W=w,Z=z,\Delta=1)$ is calculated with respect to the law of $(Y,\Delta,X,W,Z)$ conditional on $\widehat{F}_{T \mid X,Z}$ with $(Y,\Delta,X,W,Z)$ independent of the data $\{Y_i,\Delta_i,X_i,W_i,Z_i\}^n_{i=1}$. Further, let $\mathcal{R}_{X,W,Z} = \{(x,w,z) \in \mathcal{R}_X \times \mathcal{R}_W \times \mathcal{R}_Z : \mathbbm{P}(X=x,W=w,Z=z) > 0 \}$ and similarly for $\mathcal{R}_{X,Z}$. Note that $F$ always denotes the distribution function, $S=1-F$ the survival function and $f$ the density function. Moreover, we omit definitions when conditioning $V$ or $\widehat{V}$ on, or forming the joint distribution with, any subset of $(X,W,Z,\Delta)$, as these are implied by earlier definitions. Before proving the main results of Section \ref{asymptoticsection}, we provide the following technical lemmas.


\section{Technical lemmas}\label{secapplemma}

\begin{lemma}\label{Lemmadosnkerclass}
Under Assumptions \ref{Asupport} - \ref{supportass}, we have that
\begin{itemize}
    \item[(i)] $\sup_{v \in \mathcal{I}}\left\lvert \widehat{S}_{\widehat{V}}(v) - \widehat{S}_{V}(v) - S_{\widehat{V}}(v) + S_V(v) \right\rvert = o_p(n^{-1/2}), $
    \item[(ii)] $\sup_{v \in \mathcal{I}}\left\lvert \widehat{S}_{\widehat{V},1}(v) - \widehat{S}_{V,1}(v) -S_{\widehat{V},1}(v) + S_{V,1}(v) \right\rvert = o_p(n^{-1/2}),$
    \item[(iii)] $\begin{aligned}
        \sup_{v \in \mathcal{I},(x,w,z) \in\mathcal{R}_{X,W,Z}} \Big\lvert &\widehat{S}_{\widehat{V},X,W,Z}(v,x,w,z)  - \widehat{S}_{V,X,W,Z}(v,x,w,z) \\ &- S_{\widehat{V},X,W,Z}(v,x,w,z) + S_{V,X,W,Z}(v,x,w,z)  \Big\rvert = o_p(n^{-1/2}),    \end{aligned}$
    \item[(iv)] $\begin{aligned}
        \sup_{v \in \mathcal{I},(x,w,z) \in\mathcal{R}_{X,W,Z}} \Big\lvert & \widehat{S}_{\widehat{V},X,W,Z,1}(v,x,w,z)  - \widehat{S}_{V,X,W,Z,1}(v,x,w,z) \\&  - S_{\widehat{V},X,W,Z,1}(v,x,w,z) + S_{V,X,W,Z,1}(v,x,w,z)  \Big\rvert = o_p(n^{-1/2}). \end{aligned}$
\end{itemize}
\end{lemma}
\begin{proof}
We will start by giving the proof of $(i)$. Define the class
\begin{align*}
\mathcal{F} = \bigg\{& (y,x,z) \to \mathbbm{1}\left(F(y \mid x,z) \leq v \right)  : v \in \mathcal{I}, \\& \qquad F(\cdot \mid x,z) \text{ is  monotone onto } [0,1] \text{ for all } (x,z) \in \mathcal{R}_{X,Z}  \bigg\}.
\end{align*}
We will now show that $\mathcal{F}$ is Donsker by proving that
 $$
  \int^\infty_0  \sqrt{\log N_{[\hspace{0.05cm}]}\big(\epsilon, \mathcal{F}, \lVert\cdot\rVert_{L_2}\big)}\diff \epsilon < \infty,
  $$
where $N_{[\hspace{0.05cm}]}\big(\epsilon, \mathcal{F}, \lVert\cdot\rVert_{L_2}\big)$ is the bracketing number and $\lVert g \rVert_{L_j} = \left(\EX[g^j(Y,X,Z)]\right)^{1/j}$ for any $j$ and function $g$ \citep[Theorem 2.5.6]{van1996weak}. Since $X$ and $Z$ only take on a finite number of values by Assumption \ref{Asupport}, we can rewrite the function class
\begin{align*}
\mathcal{F} = \bigg\{&(y,x,z) \to   \sum_{(\Tilde{x},\Tilde{z}) \in \mathcal{R}_{X,Z}} \mathbbm{1}\left(x=\Tilde{x},z=\Tilde{z}\right)\mathbbm{1}\left(F(y \mid \Tilde{x},\Tilde{z}) \leq v \right)  : v \in \mathcal{I}, \\ & \qquad  F(\cdot \mid \Tilde{x},\Tilde{z}) \text{ is  monotone onto } [0,1] \text{ for all } (\Tilde{x},\Tilde{z}) \in \mathcal{R}_{X,Z}  \bigg\},
\end{align*}
which is the sum of the classes
$$
\mathcal{F}_{\Tilde{x},\Tilde{z}}  = \bigg\{(y,x,z) \to \mathbbm{1}\left(x=\Tilde{x},z=\Tilde{z} \right) \mathbbm{1}\left(F(y \mid \Tilde{x},\Tilde{z} ) \leq v \right)  : v \in \mathcal{I}, F(\cdot \mid \Tilde{x},\Tilde{z}) \text{ is  monotone onto } [0,1]  \bigg\},
$$
over all $(\Tilde{x},\Tilde{z}) \in \mathcal{R}_{X,Z}$.
Note that we can rewrite $\mathcal{F}_{\Tilde{x},\Tilde{z}}$ as the product of the classes
$$
\mathcal{G}_{\Tilde{x},\Tilde{z}} = \bigg\{(x,z) \to \mathbbm{1}\left(x=\Tilde{x},z=\Tilde{z} \right)  \bigg\},
$$
and 
$$
\mathcal{H}_{\Tilde{x},\Tilde{z}} = \bigg\{y \to \mathbbm{1}\left(F(y \mid \Tilde{x},\Tilde{z} ) \leq v \right)  : v \in \mathcal{I}, F(\cdot \mid \Tilde{x},\Tilde{z}) \text{ is  monotone onto } [0,1]  \bigg\}.
$$
It is clear that the bracketing number of $\mathcal{G}_{\Tilde{x},\Tilde{z}}$ is 1 and that $\mathcal{H}_{\Tilde{x},\Tilde{z}}$ has VC-dimension 2 for all $(\Tilde{x},\Tilde{z}) \in \mathcal{R}_{X,Z}$ \citep[Example 2.6.1]{van1996weak}, meaning they are both Donsker. Since both of these classes are uniformly bounded by 1, we know from Lemma 9.25 by \cite{kosorok2008introduction} that $N_{[\hspace{0.05cm}]}\big(\epsilon, \mathcal{F}_{\Tilde{x},\Tilde{z}}, \lVert\cdot\rVert_{L_2}\big) = O(\epsilon^{-2})$. Moreover, this also implies that $$N_{[\hspace{0.05cm}]}\big(\epsilon, \mathcal{F}, \lVert\cdot\rVert_{L_2}\big) = O(\epsilon^{-2d_{x,z}}),$$ 
where $d_{x,z}=\lvert\mathcal{R}_{X,Z}\rvert$. From this it follows that $\mathcal{F}$ is Donsker since
$$
  \int^1_0  \sqrt{\log(1/\epsilon)}\diff \epsilon < \infty,
$$
and for $\epsilon>1$, it is clear that 1 bracket suffices. We continue by using Corollary 2.3.12 in \cite{van1996weak}, from which we know that for every $\epsilon >0$:
\begin{align*}
\lim_{\kappa \downarrow 0} \lim_{n \to \infty} \mathbbm{P}\Bigg(\sup_{f,g \in \mathcal{F}, \rho(f-g) < \kappa}\bigg\lvert   n^{-1/2} \sum^n_{i=1} \big\{ &f(Y_i,X_i,Z_i) -g(Y_i,X_i,Z_i) \\&  - \EX[f(Y,X,Z)]  + \EX[g(Y,X,Z)] \big\} \bigg\rvert > \epsilon\Bigg)=0,
\end{align*}
where $\mathcal{F}$ is a Donsker class and $$\rho^2(f-g)=\EX\left[\left\{ f(Y,X,Z)-g(Y,X,Z)-\EX\left[f(Y,X,Z)-g(Y,X,Z)\right]\right\}^2\right].$$
Therefore, we calculate
\begin{align*}
& \rho^2\left(\mathbbm{1}(\widehat{F}_{T\mid X,Z}(Y \mid X,Z) \leq v ) - \mathbbm{1}(F_{T\mid X,Z}(Y \mid X,Z) \leq v)\right) \\ & \leq \EX\left[\left(\mathbbm{1}(\widehat{F}_{T\mid X,Z}(Y \mid X,Z) \leq v ) - \mathbbm{1}(F_{T\mid X,Z}(Y \mid X,Z) \leq v)\right)^2 \right] \\ & = \int \Big[F_{Y\mid X,Z}\left(\widehat{F}_{T \mid X,Z}^{-1}(v \mid x,z) \mid x,z \right)dF_{X,Z}(x,z)  \\ & \quad + \int F_{Y\mid X,Z}\left(F_{T \mid X,Z}^{-1}(v \mid x,z) \mid x,z \right)dF_{X,Z}(x,z) \\ & \quad - 2\int F_{Y\mid X,Z}\left(\widehat{F}_{T \mid X,Z}^{-1}(v \mid x,z) \wedge F_{T \mid X,Z}^{-1}(v \mid x,z) \mid x,z \right) dF_{X,Z}(x,z) \\ & =\int \left\lvert F_{Y\mid X,Z}\left(\widehat{F}_{T \mid X,Z}^{-1}(v \mid x,z) \mid x,z \right)- F_{Y\mid X,Z}\left(F_{T \mid X,Z}^{-1}(v \mid x,z) \mid x,z \right) \right\rvert dF_{X,Z}(x,z) \\ & \leq M \int \left\lvert \widehat{F}_{T \mid X,Z}^{-1}(v \mid x,z)- F_{T \mid X,Z}^{-1}(v \mid x,z) \right\rvert dF_{X,Z}(x,z) \\ & \leq M \max_{(x,z) \in \mathcal{R}_{X,Z}}\sup_{v\in \mathcal{I}} \left\lvert \widehat{F}_{T \mid X,Z}^{-1}(v \mid x,z)- F_{T \mid X,Z}^{-1}(v \mid x,z) \right\rvert,
\end{align*}
with $M$ some constant and where $\wedge$ indicates the minimum. It now follows from Theorem 2 by \cite{lo1986product} that $$\sup_{v\in \mathcal{I}} \left\lvert \widehat{F}_{T \mid X,Z}^{-1}(v \mid x,z)- F_{T \mid X,Z}^{-1}(v \mid x,z) \right\rvert =o(1) \quad a.s.,$$ for all $ (x,z) \in \mathcal{R}_{X,Z}$, which completes the proof. The proofs of $(ii)$, $(iii)$ and $(iv)$ are similar to the proof of $(i)$ and therefore omitted.
\end{proof}

\begin{lemma}\label{lemmabigO}
    Under Assumptions \ref{Asupport} - \ref{supportass}, we have that
    \begin{itemize}
        \item[(i)] $
    \sup_{v \in \mathcal{I}}\left\lvert\widehat{S}_{\widehat{V}}(v)-S_V(v)\right\rvert =O_p(n^{-1/2}),
    $ 
    \item[(ii)]  $
    \sup_{v \in \mathcal{I}, (x,w,z) \in \mathcal{R}_{X,W,Z}}\left\lvert\widehat{S}_{\widehat{V}\mid X,W,Z}(v\mid x,w,z)-S_{V \mid X,W,Z}(v \mid x,w,z)\right\rvert =O_p(n^{-1/2}).
    $
    \end{itemize}

\end{lemma}
\begin{proof}
    We will start by giving the proof of $(i)$. Using Lemma \ref{Lemmadosnkerclass}, it is clear that $$ \widehat{S}_{\widehat{V}}(v)-S_V(v) =  \widehat{S}_{V}(v)-S_V(v)+S_{\widehat{V}}(v)-S_V(v) + o_p(n^{-1/2}),$$
uniformly in $v \in \mathcal{I}$. Note that by using a Taylor expansion, with $S_{Y\mid X,Z}(y\mid x,z) = \mathbbm{P}(Y > y \mid X=x,Z =z)$ and $f_{Y\mid X,Z}$ the corresponding density, we have that
\begin{align*}
S_{\widehat{V}}(v) & =  \sum_{(x,z) \in \mathcal{R}_{X,Z}}  S_{Y\mid X,Z} (\widehat{F}_{T \mid X,Z}^{-1}(v \mid x,z)\mid x,z)p_{X,Z}(x,z) \\ & =   \sum_{(x,z) \in \mathcal{R}_{X,Z}} \bigg[ S_{Y\mid X,Z}(F_{T \mid X,Z}^{-1}(v \mid x,z)\mid x,z)  - f_{Y\mid X,Z}(F_{T \mid X,Z}^{-1}(v \mid x,z)\mid x,z)\\ & \quad \times \big[\widehat{F}_{T \mid X,Z}^{-1}(v \mid x,z) - F_{T \mid X,Z}^{-1}(v \mid x,z)\big] \bigg]p_{X,Z}(x,z) + o_p(n^{-1/2}),
\end{align*}
such that
\begin{align*}
& S_{\widehat{V}}(v)-S_V(v)  \\ & \quad =  \sum_{(x,z) \in \mathcal{R}_{X,Z}} f_{Y\mid X,Z}(F_{T \mid X,Z}^{-1}(v \mid x,z)\mid x,z) \left[F_{T \mid X,Z}^{-1}(v \mid x,z)-\widehat{F}_{T \mid X,Z}^{-1}(v \mid x,z)\right]p_{X,Z}(x,z) \\ & \qquad + o_p(n^{-1/2}) \\ & \quad = \sum_{(x,z) \in \mathcal{R}_{X,Z}}  \frac{f_{Y\mid X,Z}(F_{T \mid X,Z}^{-1}(v \mid x,z)\mid x,z)}{f_{T\mid X,Z}(F_{T \mid X,Z}^{-1}(v \mid x,z)\mid x,z)} \left[ \widehat{F}_{V_T \mid X,Z}(v \mid x,z)-F_{V_T \mid X,Z}(v \mid x,z)\right]p_{X,Z}(x,z) \\ & \qquad + o_p(n^{-1/2}),
\end{align*}
uniformly in $v \in \mathcal{I}$ and where the second equality follows from Theorem 2 by \cite{lo1986product}. Using Theorem 1.2 by \cite{gill1983large}, it now follows that $\sup_{v \in \mathcal{I}}\left\lvert S_{\widehat{V}}(v)-S_V(v)\right\rvert =O_p(n^{-1/2})$. Moreover, using standard empirical process results, it is clear that $\sup_{v \in \mathcal{I}}\left\lvert \widehat{S}_{V}(v)-S_V(v)\right\rvert =O_p(n^{-1/2})$, from which the result follows. We continue with the proof of $(ii)$. Applying Lemma \ref{Lemmadosnkerclass}, we have that
\begin{align*}
& \widehat{S}_{\widehat{V}\mid X,W,Z}(v\mid x,w,z)-S_{V \mid X,W,Z}(v \mid x,w,z) \\ \quad&=  \widehat{S}_{V\mid X,W,Z}(v\mid x,w,z)-S_{V\mid X,W,Z}(v\mid x,w,z)\\& \quad +\left[ S_{\widehat{V}\mid X,W,Z}(v\mid x,w,z)-S_{V\mid X,W,Z}(v\mid x,w,z)\right]\times\left(1+o_p(1)\right)  + o_p(n^{-1/2}),\end{align*} 
uniformly in $(v,x,w,z) \in \mathcal{I}\times \mathcal{R}_{X,W,Z}$. Using a similar Taylor expansion as in the proof of $(i)$, we have that
\begin{align*}
& S_{\widehat{V}\mid X,W,Z}(v\mid x,w,z)-S_{V\mid X,W,Z}(v\mid x,w,z)  \\ & \quad = \frac{f_{Y\mid X,W,Z}(F_{T \mid X,Z}^{-1}(v \mid x,w,z)\mid x,w,z)}{f_{T\mid X,Z}(F_{T \mid X,Z}^{-1}(v \mid x,z)\mid x,z)} \left[\widehat{F}_{V_T \mid X,Z}(v \mid x,z) - F_{V_T \mid X,Z}(v \mid x,z)\right]  + o_p(n^{-1/2}),
\end{align*}
uniformly in $(v,x,w,z) \in \mathcal{I} \times \mathcal{R}_{X,W,Z}$. The result now follows from using the same empirical process arguments as those used in the proof of $(i)$.
\end{proof}

\begin{lemma}\label{lemmaremainder0}
    Under Assumptions \ref{Asupport} - \ref{supportass}, we have that
    \begin{itemize}
        \item[(i)] $
    \sup_{v \in \mathcal{I}}\left\lvert\int_0^v\left(\widehat{S}^{-1}_{\widehat{V}}(u)-S^{-1}_{V}(u)\right)d\left(\widehat{S}_{\widehat{V},1}(u)-S_{V,1}(u)\right)\right\rvert = o_p(n^{-1/2}),
    $
    \item[(ii)] $ \begin{aligned}
    \sup_{v \in \mathcal{I}, (x,w,z) \in \mathcal{R}_{X,W,Z}}\Bigg\lvert&\int_0^v\left(\widehat{S}^{-1}_{\widehat{V}\mid X,W,Z}(u\mid x,w,z)-S^{-1}_{V\mid X,W,Z}(u\mid x,w,z)\right) \\ & \qquad d\bigg(\widehat{S}_{\widehat{V},1\mid X,W,Z}(u\mid x,w,z)-S_{V,1\mid X,W,Z}(u\mid x,w,z)\bigg)\Bigg\rvert = o_p(n^{-1/2}).    
    \end{aligned}
    $
    \end{itemize}
\end{lemma}
\begin{proof}
    We will start by giving the proof of $(i)$. Following the proof of Lemma 2 by \cite{lo1986product}, we start by dividing $\mathcal{I}$ into subintervals $[l_i,l_{i+1}],$ $i=1,2,\dots,k_n$ where $k_n = O(n^{1/2}\log(n)^{-1/2})$ and $0 =l_1 <l_2<\dots < l_{k_n+1} = 1-\gamma$ are such that $S_V(l_i)-S_V(l_{i+1}) \leq Mn^{-1/2}\log(n)^{1/2}$, where $M$ is some constant whose value may change from line to line. We then have that 
    \begin{align}
        & \left\lvert\int_0^v\left({\widehat{S}^{-1}_{\widehat{V}}(u)}-{S^{-1}_{V}(u)}\right)d\left(\widehat{S}_{\widehat{V},1}(u)-S_{V,1}(u)\right)\right\rvert \nonumber \\& \quad  \leq k_n \sup_{u\in \mathcal{I}}\left\lvert {\widehat{S}^{-1}_{\widehat{V}}(u)}-{S^{-1}_{V}(u)} \right\rvert \max_{1\leq i \leq k_n}\left\lvert \widehat{S}_{\widehat{V},1}(l_{i+1}) -\widehat{S}_{\widehat{V},1}(l_{i})-S_{V,1}(l_{i+1})+S_{V,1}(l_i)\right\rvert \label{termA} \\& \quad + 2 \max_{1\leq i \leq k_n} \sup_{u\in [l_i,l_{i+1}]}\left\lvert{\widehat{S}^{-1}_{\widehat{V}}(u)} - {\widehat{S}^{-1}_{\widehat{V}}(l_i)}-{S^{-1}_{V}(u)}+{S^{-1}_{V}(l_i)}\right\rvert. \label{termB}
    \end{align}
    We first investigate \eqref{termB}, and start by further dividing each $[l_i,l_{i+1}]$ into the subintervals $[l_{ij},l_{i(j+1)}]$, $j=1,\dots,a_n$ such that $\left\lvert S_{V}(l_{i(j+1)})-S_{V}(l_{ij})\right\rvert = O(n^{-3/4}\log(n)^{3/4})$ uniformly in $i,j$, and $a_n=O(n^{1/4}\log(n)^{-1/4})$. Thanks to our Lemma \ref{lemmabigO}, it follows that
    \begin{align*}
        & \sup_{u\in [l_i,l_{i+1}]}\left\lvert{\widehat{S}^{-1}_{\widehat{V}}(u)} - {\widehat{S}^{-1}_{\widehat{V}}(l_i)}-{S^{-1}_{V}(u)}+{S^{-1}_{V}(l_i)}\right\rvert  \\ & \quad \leq \sup_{u\in [l_i,l_{i+1}]}\left\lvert\frac{\widehat{S}_{\widehat{V}}(u)-S_{V}(u)}{S_{V}(u)^2} - \frac{\widehat{S}_{\widehat{V}}(l_i)-S_{V}(l_i)}{S_{V}(l_i)^2}\right\rvert +O_p(n^{-1}\log(n))  \\ & \quad \leq \sup_{u\in [l_i,l_{i+1}]}\frac{\left\lvert \widehat{S}_{\widehat{V}}(u)-\widehat{S}_{\widehat{V}}(l_i)-S_{V}(u) +S_{V}(l_i)  \right\rvert}{S_{V}(l_{i+1})^2} +O_p(n^{-1}\log(n)) \\ & \quad \leq M\max_{1 \leq j \leq a_n}\left\lvert \widehat{S}_{\widehat{V}}(l_{ij})-\widehat{S}_{\widehat{V}}(l_i)-S_{V}(l_{ij}) +S_{V}(l_i)  \right\rvert +O_p(n^{-3/4}\log(n)^{3/4}).
    \end{align*}
    From Lemma \ref{Lemmadosnkerclass}, it follows that
    \begin{align}
        & M\max_{1 \leq j \leq a_n}\left\lvert \widehat{S}_{\widehat{V}}(l_{ij})-\widehat{S}_{\widehat{V}}(l_i)-S_{V}(l_{ij}) +S_{V}(l_i)  \right\rvert \nonumber \\ & \quad \leq M\max_{1 \leq j \leq a_n}\left\lvert \widehat{S}_{V}(l_{ij})-\widehat{S}_{V}(l_i)-S_{V}(l_{ij}) +S_{V}(l_i)  \right\rvert \label{oldtermiid} \\ & \qquad + M\max_{1 \leq j \leq a_n}\left\lvert S_{\widehat{V}}(l_{ij})-S_{\widehat{V}}(l_i)-S_{V}(l_{ij}) +S_{V}(l_i)  \right\rvert + o_p(n^{-1/2}). \label{newtermiid}
    \end{align}
    Since \eqref{oldtermiid} is the same term as the one in Lemma 2 by \cite{lo1986product}, we know that it is $O_p({n}^{-3/4}\log(n)^{3/4})$. Using the same Taylor expansion from the beginning of the proof of $(i)$ from our Lemma \ref{lemmabigO}, we have that 
      \begin{align*}
        &  M\max_{1 \leq j \leq a_n}\left\lvert S_{\widehat{V}}(l_{ij})-S_{\widehat{V}}(l_i)-S_{V}(l_{ij}) +S_{V}(l_i)  \right\rvert \\ & \leq Md_{x,z}\max_{(x,z)\in \mathcal{R}_{X,Z}}\max_{1 \leq j \leq a_n}\Big\lvert \Psi_{x,z}(l_{ij})\widehat{F}_{V_T \mid X,Z}(l_{ij} \mid x,z) - \Psi_{x,z}(l_{i})\widehat{F}_{V_T \mid X,Z}(l_{i} \mid x,z) \\ & \qquad  - \Psi_{x,z}(l_{ij})F_{V_T \mid X,Z}(l_{ij}  \mid x,z)+\Psi_{x,z}(l_{i})F_{V_T \mid X,Z}(l_{i}  \mid x,z) \Big\rvert +o_p(n^{-1/2}), 
    \end{align*}
    where $d_{x,z} = \lvert \mathcal{R}_{X,Z} \rvert$ and $\Psi_{x,z}(u)=f_{Y\mid X,Z}(F_{T \mid X,Z}^{-1}(u \mid x,z)\mid x,z)/f_{T\mid X,Z}(F_{T \mid X,Z}^{-1}(u \mid x,z)\mid x,z)$. Similar to the steps of the proof of Theorem 2 by \cite{lo1986product}, we find that
    \begin{align*}
        & \max_{1 \leq j \leq a_n}\Big\lvert  \Psi_{x,z}(l_{ij})\widehat{F}_{V_T \mid X,Z}(l_{ij} \mid x,z) - \Psi_{x,z}(l_{i})\widehat{F}_{V_T \mid X,Z}(l_{i} \mid x,z)  \\ & \qquad \qquad - \Psi_{x,z}(l_{ij})F_{V_T \mid X,Z}(l_{ij}  \mid x,z)+\Psi_{x,z}(l_{i})F_{V_T \mid X,Z}(l_{i}  \mid x,z) \Big\rvert  = O_p({n}^{-3/4}\log(n)^{3/4}).
    \end{align*}
    Since the proof of \eqref{termA} is similar, it is omitted. The proof of $(ii)$ is almost identical to the proof of $(i)$, and therefore also omitted.
\end{proof}

\begin{lemma}\label{lemmaasymptotrep}
    Under Assumptions \ref{Asupport} - \ref{supportass}, we have that

    \begin{itemize}
        \item[(i)] $
    \sup_{v \in \mathcal{I}}\left\lvert\widehat{F}_{\widehat{V}_T}(v)-v\right\rvert =o_p(n^{-1/2}),
    $
        \item[(ii)] $
    \sup_{v \in \mathcal{I},(x,w,z)\in \mathcal{R}_{X,W,Z}}\left\lvert\widehat{F}_{\widehat{V}_T\mid X,W,Z}(v\mid x,w,z)-F_{V_T\mid X,W,Z}(v\mid x,w,z)\right\rvert =O_p(n^{-1/2}).
    $
    \end{itemize}

\end{lemma}
\begin{proof}
We will start by giving the proof of $(i)$. Using Lemma \ref{lemmabigO}, we can follow the same steps from the proof of Theorem 1 by \cite{lo1986product} to show that
\begin{align}
& \log\left(1-\widehat{F}_{\widehat{V}_T}(v)\right) - \log(1-v) \nonumber \\& \quad =-\int_0^v\frac{\widehat{S}_{\widehat{V}}(u)}{S_V(u)^2}dS_{V,1}(u)+\int_0^v\frac{d\widehat{S}_{\widehat{V},1}(u)}{S_V(u)} \label{losinghiid} \\& \qquad + \int_0^v\left({\widehat{S}^{-1}_{\widehat{V}}(u)}-{S^{-1}_{V}(u)}\right)d\left(\widehat{S}_{\widehat{V},1}(u)-S_{V,1}(u)\right) + o_p(n^{-1/2}), \label{remainderlosingh}
\end{align}
uniformly in $v \in \mathcal{I}$. In view of Lemma \ref{lemmaremainder0}, it is clear that \eqref{remainderlosingh} is $o_p(n^{-1/2})$ uniformly in $v\in \mathcal{I}$. Further, using Lemma \ref{Lemmadosnkerclass}, it follows that \eqref{losinghiid} can be written as
\begin{align}
     -\int_0^v\frac{\widehat{S}_{V}(u)}{S_V(u)^2}dS_{V,1}(u)+\int_0^v\frac{d\widehat{S}_{V,1}(u)}{S_V(u)}  -\int_0^v\frac{S_{\widehat{V}}(u)}{S_V(u)^2}dS_{V,1}(u)+\int_0^v\frac{dS_{\widehat{V},1}(u)}{S_V(u)} +o_p(n^{-1/2}),\label{remaindervhat} 
\end{align}
uniformly in $v\in \mathcal{I}$. Note that the first two terms of \eqref{remaindervhat} are the same terms from Theorem 1 by \cite{lo1986product}. Using the same Taylor expansion as in the proof of Lemma \ref{lemmabigO}, we get that
\begin{align*}
S_{\widehat{V}}(v) & =  \sum_{(x,z) \in \mathcal{R}_{X,Z}} S_{Y\mid X,Z}(\widehat{F}_{T \mid X,Z}^{-1}(v \mid x,z)\mid x,z)p_{X,Z}(x,z) \\ & =   \sum_{(x,z) \in \mathcal{R}_{X,Z}} \bigg[ S_{Y\mid X,Z}(F_{T \mid X,Z}^{-1}(v \mid x,z)\mid x,z)  - f_{Y\mid X,Z}(F_{T \mid X,Z}^{-1}(v \mid x,z)\mid x,z)\\ & \quad \times \big[\widehat{F}_{T \mid X,Z}^{-1}(v \mid x,z) - F_{T \mid X,Z}^{-1}(v \mid x,z)\big] \bigg]p_{X,Z}(x,z) + o_p(n^{-1/2}),
\end{align*}
uniformly in $v \in \mathcal{I}$, and similarly for $S_{\widehat{V},1}(v)$. Therefore, we get that the last two terms of \eqref{remaindervhat} are equal to
\begin{align*}
     & \sum_{(x,z) \in \mathcal{R}_{X,Z}} p_{X,Z}(x,z)\Bigg\{\int_0^v\frac{f_{V\mid X,Z}(u\mid x,z) \big[\widehat{F}_{T \mid X,Z}^{-1}(u \mid x,z) - F_{T \mid X,Z}^{-1}(u \mid x,z)\big]}{S_{V}(u)^2}dS_{V,1}(u) \\ & \qquad -\int_0^v\frac{d\Big(f_{V,1\mid X,Z}(u\mid x,z) \big[\widehat{F}_{T \mid X,Z}^{-1}(u \mid x,z) - F_{T \mid X,Z}^{-1}(u \mid x,z)\big]\Big)}{S_{V}(u)}\Bigg\} +o_p(n^{-1/2}),
\end{align*}
uniformly in $v\in \mathcal{I}$ with $f_{V\mid X,Z}$ and $f_{V,1\mid X,Z}$ the corresponding densities of $S_{V\mid X,Z}$ and $S_{V,1\mid X,Z}$ respectively. Let $S_{V_T}(v)=\mathbbm{P}(V_T > v)$ and $S_{V_C}(v)=\mathbbm{P}(V_C > v)$. Using Theorem 2 by \cite{lo1986product} and the fact that $V_T \indep X,Z$, this can be rewritten as 
\begin{align}
    & \sum_{(x,z) \in \mathcal{R}_{X,Z}} p_{X,Z}(x,z)\Bigg\{\int_0^v\frac{f_{V\mid X,Z}(u\mid x,z) \big[\widehat{F}_{V_T \mid X,Z}(u \mid x,z) - F_{V_T \mid X,Z}(u \mid x,z)\big]}{S_{V}(u)S_{V_T}(u)}du \nonumber \\ &  \qquad +\int_0^v\frac{d\Big(S_{V_C\mid X,Z}(u\mid x,z) \big[\widehat{F}_{V_T \mid X,Z}(u \mid x,z) - F_{V_T \mid X,Z}(u \mid x,z)\big]\Big)}{S_{V}(u)}\Bigg\} +o_p(n^{-1/2}), \label{pluggedinrem}
\end{align}
uniformly in $v \in \mathcal{I}$. Note that we have also used that $dS_{V,1}(u)=-f_{V_T}(u)S_{V_C}(u)du$ with $f_{V_T}$ the corresponding density of $S_{V_T}$. Further, let $N_{i,x,z}(v) = \mathbbm{1}(V_i\geq v,X_i=x,Z_i = z)$ and $N^1_{i,x,z}(v)=\mathbbm{1}(V_i\geq v, \Delta_i = 1,X_i=x,Z_i = z)$. Adapting Theorem 1 by \cite{lo1986product} to the conditional Kaplan-Meier, we have that
\begin{align*}
    \widehat{F}_{V_T \mid X,Z}(u \mid x,z) - F_{V_T \mid X,Z}(u \mid x,z) & = \frac{S_{V_T\mid X,Z}(u\mid x,z)}{\sum_{i=1}^n\mathbbm{1}(X_i=x,Z_i=z)}\sum_{i=1}^n\xi_i(u,x,z) + o_p(n^{-1/2}),
\end{align*}
uniformly in $v \in \mathcal{I}$ with $$\xi_i(u,x,z) = \int^u_0\frac{N_{i,x,z}(s)dS_{V,1\mid X,Z}(s\mid x,z)}{S_{V\mid X,Z}(s\mid x,z)^2}-\int^u_0\frac{dN^1_{i,x,z}(s)}{S_{V\mid X,Z}(s\mid x,z)}.$$ Therefore, using again that $V_T \indep X,Z$, we have that \eqref{pluggedinrem} is equal to
\begin{align*}
     n^{-1} \sum_{(x,z) \in \mathcal{R}_{X,Z}}  \Bigg\{&\int_0^v\frac{f_{V\mid X,Z}(u\mid x,z)   \sum_{i=1}^n \xi_i(u,x,z) }{S_{V}(u)}du \\ &  \quad +\int_0^v\frac{d\Big(S_{V\mid X,Z}(u\mid x,z)  \sum_{i=1}^n \xi_i(u,x,z)\Big)}{S_{V}(u)}\Bigg\} +o_p(n^{-1/2}),
\end{align*}
which can be simplified to
\begin{align*}
    & n^{-1}\sum_{(x,z) \in \mathcal{R}_{X,Z}}  \Bigg\{\int_0^v\frac{S_{V\mid X,Z}(u\mid x,z) \sum_{i=1}^n d\left(\xi_i(u,x,z)\right)}{S_{V}(u)}\Bigg\} +o_p(n^{-1/2}).
\end{align*}
Clearly, this is equal to
\begin{align}
    &  n^{-1} \sum_{i=1}^n \sum_{(x,z) \in \mathcal{R}_{X,Z}} \Bigg\{\int_0^v  \frac{N_{i,x,z}(u)dS_{V,1\mid X,Z}(u\mid x,z)}{S_{V\mid X,Z}(u\mid x,z)S_V(u)}-\frac{dN^1_{i,x,z}(u)}{S_{V}(u)}\Bigg\} +o_p(n^{-1/2}). \label{almost there}
\end{align}
Note that $$dS_{V,1\mid X,Z}(u\mid x,z) = -f_{V,1\mid X,Z}(u\mid x,z)du  = -f_{V,1}(u)\frac{S_{V_C\mid X,Z}(u\mid x,z)}{S_{V_C}(u)}du,$$
such that \eqref{almost there} is equal to
\begin{align*}
    & n^{-1}  \sum_{i=1}^n \sum_{(x,z) \in \mathcal{R}_{X,Z}} \Bigg\{\int_0^v  \frac{N_{i,x,z}(u)dS_{V,1}(u)}{S_V(u)^2}-\frac{dN^1_{i,x,z}(u)}{S_{V}(u)}\Bigg\} +o_p(n^{-1/2}),
\end{align*}
which exactly compensates the first two terms of  \eqref{remaindervhat}, meaning that $$\log\left(1-\widehat{F}_{\widehat{V}_T}(v)\right) - \log(1-v) = o_p(n^{-1/2}),$$ uniformly in $v \in \mathcal{I}$. Using a Taylor expansion, the result follows. We continue with the proof of $(ii)$. Using Lemmas \ref{lemmabigO} and \ref{lemmaremainder0}, we can follow the same steps as the proof of $(i)$ to show that
\begin{align*}
& \log\left(1-\widehat{F}_{\widehat{V}_T\mid X,W,Z}(v\mid x,w,z)\right) - \log\left(1-F_{V_T\mid X,W,Z}(v\mid x,w,z)\right) \nonumber \\& \quad =-\int_0^v\frac{\widehat{S}_{\widehat{V}\mid X,W,Z}(u\mid x,w,z)}{S_{V\mid X,W,Z}(u\mid x,w,z)^2}dS_{V,1\mid X,W,Z}(u\mid x,w,z)+\int_0^v\frac{d\widehat{S}_{\widehat{V},1\mid X,W,Z}(u\mid x,w,z)}{S_{V\mid X,W,Z}(u\mid x,w,z)} + o_p(n^{-1/2}), 
\end{align*}
 uniformly in $v,x,w$ and $z$. 
 Using Lemma \ref{Lemmadosnkerclass}, it follows that this can be written as
    \begin{align*}
     & \widehat{p}^{-1}_{X,W,Z}(x,w,z)\Bigg\{-\int_0^v\frac{\widehat{S}_{V,X,W,Z}(u,x,w,z)}{S_{V\mid X,W,Z}(u \mid x,w,z)^2}dS_{V,1\mid X,W,Z}(u \mid x,w,z) \\ & \quad +\int_0^v\frac{d\widehat{S}_{V,X,W,Z,1}(u,x,w,z)}{S_{V\mid X,W,Z}(u \mid x,w,z)}  -\int_0^v\frac{S_{\widehat{V},X,W,Z}(u,x,w,z)}{S_{V\mid X,W,Z}(u \mid x,w,z)^2}dS_{V,1\mid X,W,Z}(u \mid x,w,z)\\ & \quad +\int_0^v\frac{dS_{\widehat{V},X,W,Z,1}(u , x,w,z)}{S_{V\mid X,W,Z}(u \mid x,w,z)}\Bigg\} +o_p(n^{-1/2}),
\end{align*}
uniformly in $(v,x,w,z) \in \mathcal{I} \times \mathcal{R}_{X,W,Z}$. Using a similar Taylor expansion as in the proof of $(i)$ and Theorem 1 by \cite{lo1986product}, it follows that

  \begin{align*}
  & \log\left(1-\widehat{F}_{\widehat{V}_T\mid X,W,Z}(v\mid x,w,z)\right) - \log\left(1-F_{V_T\mid X,W,Z}(v\mid x,w,z)\right)  
     \\ & \quad = \log\left(1-\widehat{F}_{V_T\mid X,W,Z}(v\mid x,w,z)\right)-\log\left(1-F_{V_T\mid X,W,Z}(v\mid x,w,z)\right)  \\ & \qquad + \Bigg\{  \int_0^v\frac{\frac{f_{V \mid X,W,Z}(u\mid x,w,z)}{f_{V_T\mid X,Z} (u\mid x,z)}\big[\widehat{F}_{V_T \mid X,Z}(u \mid x,z) - F_{V_T \mid X,Z}(u \mid x,z)\big]}{S_{V\mid X,W,Z}(u\mid x,w,z)^2} f_{V,1\mid X,W,Z}(u \mid x,w,z)du \\ & \qquad +\int_0^v\frac{d\left(\frac{f_{V,1 \mid X,W,Z}(u\mid x,w,z)}{f_{V_T\mid X,Z} (u\mid x,z)}\big[\widehat{F}_{V_T \mid X,Z}(u \mid x,z) - F_{V_T \mid X,Z}(u \mid x,z)\big]\right)}{S_{V\mid X,W,Z}(u \mid x,w,z)}\Bigg\} \times \left(1+o_p(1)\right) +o_p(n^{-1/2}),
\end{align*}
uniformly in $(v,x,w,z) \in \mathcal{I} \times \mathcal{R}_{X,W,Z}$. Under Assumption \ref{AindepTC}, this is equal to
  \begin{align*}
  &  \log\left(1-\widehat{F}_{V_T\mid X,W,Z}(v\mid x,w,z)\right)-\log\left(1-F_{V_T\mid X,W,Z}(v\mid x,w,z)\right)  \\ & \quad + \left\{\frac{f_{V_T\mid X,W,Z}(v\mid x,w,z)\big[\widehat{F}_{V_T \mid X,Z}(v \mid x,z) - F_{V_T \mid X,Z}(v \mid x,z)\big]}{f_{V_T}(v)S_{V_T\mid X,W,Z}(v \mid x,w,z)} \right\}\times\left(1+o_p(1)\right)+o_p(n^{-1/2}),
\end{align*}
uniformly in $(v,x,w,z) \in \mathcal{I} \times \mathcal{R}_{X,W,Z}$. Using Theorem 1.2 by \cite{gill1983large} and following similar steps to the proof of $(i)$, we now have that
  \begin{align*}
  & \log\left(1-\widehat{F}_{\widehat{V}_T\mid X,W,Z}(v\mid x,w,z)\right) - \log\left(1-F_{V_T\mid X,W,Z}(v\mid x,w,z)\right)  
     \\ & \quad =  \log\left(1-\widehat{F}_{V_T\mid X,W,Z}(v\mid x,w,z)\right)-\log\left(1-F_{V_T\mid X,W,Z}(v\mid x,w,z)\right)  \\ & \qquad + n^{-1}\sum_{i=1}^n\frac{f_{V_T\mid X,W,Z}(v\mid x,w,z)S_{V_T}(v)\xi_i(v,x,z)}{f_{V_T}(v)S_{V_T\mid X,W,Z}(v \mid x,w,z)p_{X,Z}(x,z)} +o_p(n^{-1/2}),
\end{align*}
uniformly in $(v,x,w,z) \in \mathcal{I} \times \mathcal{R}_{X,W,Z}$ with $\xi_i(v,x,z)$ defined in the proof of $(i)$. Using Theorem 1 by \cite{lo1986product}, we get that
  \begin{align*}
  & \log\left(1-\widehat{F}_{\widehat{V}_T\mid X,W,Z}(v\mid x,w,z)\right) - \log\left(1-F_{V_T\mid X,W,Z}(v\mid x,w,z)\right)  
     \\ & \quad =  n^{-1}\sum_{i=1}^n \left\{\frac{f_{V_T\mid X,W,Z}(v\mid x,w,z)S_{V_T}(v)\xi_i(v,x,z)}{f_{V_T}(v)S_{V_T\mid X,W,Z}(v \mid x,w,z)p_{X,Z}(x,z)}  - \frac{\zeta_i(v,x,w,z)}{p_{X,W,Z}(x,w,z)}\right\}\\ & \qquad +o_p(n^{-1/2}),
\end{align*}
uniformly in $(v,x,w,z) \in \mathcal{I} \times \mathcal{R}_{X,W,Z}$, where $$\zeta_i(v,x,w,z) = \int^v_0\frac{N_{i,x,w,z}(u)dS_{V,1\mid X,W,Z}(u\mid x,w,z)}{S_{V\mid X,W,Z}(u\mid x,w,z)^2}-\int^v_0\frac{dN^1_{i,x,w,z}(u)}{S_{V\mid X,W,Z}(u\mid x,w,z)},$$
with $N_{i,x,w,z}(v)= \mathbbm{1}(V_i\geq v,X_i=x,W_i = w,Z_i = z)$ and $N^1_{i,x,w,z}(v)=\mathbbm{1}(V_i\geq v, \Delta_i = 1,X_i=x,W_i=w,Z_i = z)$. Using a Taylor expansion, we have that
\begin{align*}
    & \widehat{F}_{\widehat{V}_T\mid X,W,Z}(v\mid x,w,z)-F_{V_T\mid X,W,Z}(v\mid x,w,z) \\ & \quad =  n^{-1}\sum_{i=1}^n  \left\{\frac{S_{V_T\mid X,W,Z}(v\mid x,w,z)\zeta_i(v,x,w,z)}{p_{X,W,Z}(x,w,z)}-\frac{f_{V_T\mid X,W,Z}(v\mid x,w,z)S_{V_T}(v)\xi_{i}(v,x,z)}{f_{V_T}(v)p_{X,Z}(x,z)}\right\} \\ & \qquad  +o_p(n^{-1/2}),
\end{align*}
uniformly in $(v,x,w,z) \in \mathcal{I} \times \mathcal{R}_{X,W,Z}$. Note that this can be rewritten as 
\begin{align*}
    & \widehat{F}_{V_T\mid X,W,Z}(v\mid x,w,z)-F_{V_T\mid X,W,Z}(v\mid x,w,z)  \\ & \quad -\frac{f_{V_T\mid X,W,Z}(v\mid x,w,z)}{f_{V_T}(v)}\left[\widehat{F}_{V_T\mid X,Z}(v\mid x,z)-F_{V_T\mid X,Z}(v\mid x,z)\right]   +o_p(n^{-1/2}).
\end{align*}
Using again Theorem 1.2 by \cite{gill1983large}, the proof is completed.
\end{proof}

\section{Main proofs}
Using the technical results obtained in Appendix \ref{secapplemma}, we are now ready to prove the main theorems stated in Section \ref{asymptoticsection}.

\subsection{Proof of Theorem 1}
    We start by noticing that 
    \begin{align}
    & \widehat{p}_{Z\mid X,W}(z\mid x,w)-p_{Z\mid X,W}(z\mid x,w) \nonumber \\ & = n^{-1}\sum_{i=1}^n\frac{\mathbbm{1}(X_i=x,W_i=w)}{p_{X,W}(x,w)}\left[\mathbbm{1}(Z_i=z)-p_{Z\mid X,W}(z\mid x,w)\right] + o_p(n^{-1/2}),\label{linearpzxw}
    \end{align}
    uniformly in $(x,w,z) \in \mathcal{R}_{X,W,Z}$. Following similar reasoning as the proof of Theorem 1 by \cite{akritas2003estimation}, we have that
    \begin{align*}
        & \widehat{F}_{\widehat{V}_T\mid X,W}(v\mid x, w) - F_{V_T\mid X,W}(v\mid x, w) \\ & \quad = \sum_{z \in\mathcal{R}_Z}\bigg\{\left[\widehat{F}_{\widehat{V}_T\mid X,W,Z}(v\mid x,w,z)-F_{V_T\mid X,W,Z}(v\mid x,w,z)\right]p_{Z\mid X,W}(z\mid x,w) \\ & \qquad + F_{V_T\mid X,W,Z}(v\mid x,w,z)\left[\widehat{p}_{Z\mid X,W}(z\mid x,w)-p_{Z\mid X,W}(z\mid x,w)\right]\bigg\} + o_p(n^{-1/2}),
    \end{align*}
    uniformly in $(v,x,w) \in \mathcal{I} \times \mathcal{R}_{X,W}$. Using Lemma \ref{lemmaasymptotrep} and \eqref{linearpzxw}, we get that
    \begin{align*}
        & \widehat{F}_{\widehat{V}_T\mid X,W}(v\mid x, w) - F_{V_T\mid X,W}(v\mid x, w)  = n^{-1}\sum_{i=1}^n \frac{\omega_i(v,x,w)}{p_{X,W}(x,w)} + o_p(n^{-1/2}), 
    \end{align*}
    uniformly in $(v,x,w) \in \mathcal{I} \times \mathcal{R}_{X,W}$, where
    \begin{align*}
         \omega_i(v,x,w) & = S_{V_T\mid X,W,Z}(v\mid x,w,Z_i)\zeta_i(v,x,w,Z_i)-p_{W\mid V_T,X,Z}(w\mid v,x,Z_i)S_{V_T}(v)\xi_{i}(v,x,Z_i)  \\ & \quad +\mathbbm{1}(X_i=x,W_i=w)\left[F_{V_T\mid X,W,Z}(v\mid x,w,Z_i)-F_{V_T\mid X,W}(v\mid x,w)\right].
    \end{align*}
    Using Lemma \ref{lemmaasymptotrep} again, we now have that
     \begin{align*}
     \widehat{F}_{\widehat{V}_T\mid X,W}(v\mid x,w) - \widehat{F}_{\widehat{V}_T}(v) = n^{-1}\sum_{i=1}^n \frac{\omega_i(v,x,w)}{p_{X,W}(x,w)} + F_{V_T\mid X,W}(v\mid x,w) -v +o_p(n^{-1/2}),  
    \end{align*}
    uniformly in $(v,x,w) \in \mathcal{I} \times \mathcal{R}_{X,W}$.
\QEDB
\subsection{Proof of Theorem 2}
To show the weak convergence of the process $$\left\{n^{1/2}\left[\widehat{F}_{\widehat{V}_T\mid X,W}(v\mid x,w) - \widehat{F}_{\widehat{V}_T}(v)\right] : (v,x,w) \in \mathcal{I} \times \mathcal{R}_{X,W} \right\},$$ under $H_0$, we will show that the function class
\begin{align*}
\mathcal{M}  = & \Big\{ (V,\Delta,X,W,Z) \to \frac{S_{V_T\mid X,W,Z}(v\mid x,w,Z)\zeta(v,x,w,V,\Delta,X,W,Z)}{p_{X,W}(x,w)} \\ & \quad  -\frac{p_{W\mid V_T,X,Z}(w\mid v,x,Z)(1-v)\xi(v,x,V,\Delta,X,Z)}{p_{X,W}(x,w)}   \\ & \quad +\frac{\mathbbm{1}(X=x,W=w)}{p_{X,W}(x,w)}\left[F_{V_T\mid X,W,Z}(v\mid x,w,Z)-v\right] : (v,x,w) \in \mathcal{I} \times \mathcal{R}_{X,W}\Big\},
\end{align*}
is Donsker, where
\begin{align*}
    \zeta(v,x,w,V,\Delta,X,W,Z) & = \int^v_0\frac{\mathbbm{1}(V\geq u,X=x,W=w)dS_{V,1\mid X,W,Z}(u\mid x,w,Z)}{S_{V\mid X,W,Z}(u\mid x,w,Z)^2} \\ & \quad +\frac{\mathbbm{1}(V\leq v, \Delta = 1,X=x,W = w)}{S_{V\mid X,W,Z}(V\mid x,w,Z)},
\end{align*}
and
\begin{align*}
    \xi(v,x,V,\Delta,X,Z) & = \int^v_0\frac{\mathbbm{1}(V\geq u,X=x)dS_{V,1\mid X,Z}(u\mid x,Z)}{S_{V\mid X,Z}(u\mid x,Z)^2} +\frac{\mathbbm{1}(V\leq v, \Delta = 1,X=x)}{S_{V\mid X,Z}(V\mid x,Z)}.
\end{align*}
Note that 
$$
\mathcal{M} = \bigcup_{(\tilde{x},\tilde{w}) \in \mathcal{R}_{X,W}} \mathcal{M}^{\tilde{x},\tilde{w}},
$$
where
\begin{align*}
\mathcal{M}^{\tilde{x},\tilde{w}} =& \Big\{ (V,\Delta,X,W,Z) \to \frac{S_{V_T\mid X,W,Z}(v\mid \tilde{x},\tilde{w},Z)\zeta(v,\tilde{x},\tilde{w},V,\Delta,X,W,Z)}{p_{X,W}(\tilde{x},\tilde{w})} \\ & \quad -\frac{p_{W\mid V_T,X,Z}(\tilde{w}\mid v,\tilde{x},Z)(1-v)\xi(v,\tilde{x},V,\Delta,X,Z)}{p_{X,W}(\tilde{x},\tilde{w})}   \\ & \quad +\frac{\mathbbm{1}(X=\tilde{x},W=\tilde{w})}{p_{X,W}(\tilde{x},\tilde{w})}\left[F_{V_T\mid X,W,Z}(v\mid \tilde{x},\tilde{w},Z)-v\right] : v \in \mathcal{I}\Big\}.
\end{align*}
Since the finite union of Donsker classes is again Donsker, it is sufficient to show that $\mathcal{M}^{\tilde{x},\tilde{w}}$ is Donsker for all $(\tilde{x},\tilde{w}) \in \mathcal{R}_{X,W}$. Using the same reasoning as in the proof of Lemma \ref{Lemmadosnkerclass}, $\mathcal{M}^{\tilde{x},\tilde{w}}$ is Donsker if 
\begin{align*}
\mathcal{M}^{\tilde{x},\tilde{w}}_{\tilde{X},\tilde{W},\tilde{Z}} =& \Big\{ (V,\Delta) \to \frac{S_{V_T\mid X,W,Z}(v\mid \tilde{x},\tilde{w},\tilde{Z})\zeta(v,\tilde{x},\tilde{w},V,\Delta,\tilde{X},\tilde{W},\tilde{Z})}{p_{X,W}(\tilde{x},\tilde{w})} \\ & \quad -\frac{p_{W\mid V_T,X,Z}(\tilde{w}\mid v,\tilde{x},\tilde{Z})(1-v)\xi(v,\tilde{x},V,\Delta,\tilde{X},\tilde{Z})}{p_{X,W}(\tilde{x},\tilde{w})}  \\ & \quad +\frac{\mathbbm{1}(\tilde{X}=\tilde{x},\tilde{W}=\tilde{w})}{p_{X,W}(\tilde{x},\tilde{w})}\left[F_{V_T\mid X,W,Z}(v\mid \tilde{x},\tilde{w},\tilde{Z})-v\right] : v \in \mathcal{I}\Big\},
\end{align*}
is Donsker for all $(\tilde{X},\tilde{W},\tilde{Z}) \in \mathcal{R}_{X,W,Z}$. Noticing that for all $(\tilde{x},\tilde{w},\tilde{X},\tilde{W},\tilde{Z})  \in \mathcal{R}_{X,W} \times \mathcal{R}_{X,W,Z}$ we have that $\mathcal{M}^{\tilde{x},\tilde{w}}_{\tilde{X},\tilde{W},\tilde{Z}}$ is a linear combination of functions that are monotone and uniformly bounded in $v \in \mathcal{I}$, the weak convergence follows from Theorem 2.7.5 by \cite{van1996weak} and Lemma 9.25 by \cite{kosorok2008introduction}. The convergence of $T_n^{KS}$ follows from the continuous mapping theorem. Moreover, the convergence of $T_n^{CM}$ follows from Assumption \ref{ascoonvpi} and using similar arguments as in the proof of Corollary 3.1 by \cite{feve2018estimation}.
\QEDB

\end{document}